\documentclass[10pt,conference]{IEEEtran}
\IEEEoverridecommandlockouts
\usepackage{cite}
\usepackage[font=small]{caption}
\usepackage{amsmath,amssymb,amsfonts}
\usepackage{subcaption}
\usepackage{graphicx}
\usepackage{textcomp}
\usepackage{algorithm}
\usepackage{algpseudocode}
\usepackage{mwe}
\usepackage{makecell}
\usepackage{textcomp}
\usepackage{comment}
\usepackage{float}
\usepackage{soul}
\usepackage{hyperref}
\usepackage{xcolor}
\usepackage{amsthm}
\usepackage{wrapfig,lipsum,booktabs}
\usepackage{multirow,multicol}
\newtheorem{definition}{Definition}

\newtheorem{lemma}{Lemma}

\newtheorem{example}{Example}

\def\BibTeX{{\rm B\kern-.05em{\sc i\kern-.025em b}\kern-.08em
    T\kern-.1667em\lower.7ex\hbox{E}\kern-.125emX}}

\newcommand{\sd}[1]{\textcolor{red}{#1}}

\def\BibTeX{{\rm B\kern-.05em{\sc i\kern-.025em b}\kern-.08em
    T\kern-.1667em\lower.7ex\hbox{E}\kern-.125emX}}

\begin{document}
\title{Mitigating Timing-Based Attacks in Real-Time Cyber-Physical Systems}
\author{Arkaprava Sain, Sunandan Adhikary, Soumyajit Dey\\
Department of Computer Science and Engineering, Indian Institute of Technology, Kharagpur, India
}


	
\maketitle
\IEEEpeerreviewmaketitle
\begin{abstract} 
\par Real-time cyber-physical systems depend on deterministic task execution to guarantee safety and correctness. Unfortunately, this determinism can unintentionally expose timing information that enables adversaries to infer task execution patterns and carry out timing-based attacks targeting safety-critical control tasks. While prior defenses aim to obscure schedules through randomization or isolation, they typically neglect the implications of such modifications on closed-loop control behavior and real-time feasibility.

\par This work studies the problem of securing real-time control workloads against timing inference attacks while explicitly accounting for both schedulability constraints and control performance requirements. We present a scheduling-based mitigation approach that introduces bounded timing perturbations to control task executions in a structured manner, reducing adversarial opportunities without violating real-time guarantees. The framework jointly considers worst-case execution behavior and the impact of execution delays on control performance, enabling the system to operate within predefined safety and performance limits. Through experimental evaluation on representative task sets and control scenarios, the proposed approach demonstrates that exposure to timing-based attacks can be significantly reduced while preserving predictable execution and acceptable control quality. 
\end{abstract}


\section{Introduction}
\label{Introduction}
Cyber Physical Systems used in various domains, such as automotive, smart grids, avionics, Industrial Internet of Things (IIoT), medical devices, etc are real time in nature. Most of these real time systems (RTS) possess \emph{safety-critical} properties, i.e any delay in their operation or failure during run-time can result in severe consequences, such as loss of life or damage to the environment/infrastructure. 
Over the past decade, the security of RTS has become very important for the following reasons: (i)  Deployment of mixed-criticality applications~\cite{giannopoulou2013scheduling}, where tasks of multiple safety and security requirements share the same computing platform; (ii) Constraints on computational resources (e.g., CPU, memory, energy) limit the use of heavy duty cryptographic data authentication and integrity verification mechanisms; (iii) Incorporation of \emph{commercial off-the-shelf} components sourced from various third-party vendors/OEMS; (iv) Connectivity to external networks, including the internet, which introduces new vulnerabilities through remote access and frequent software updates~\cite{humayed2017cyber,brooks2008automotive}. While these have improved the functionality of applications and reduced deployment costs, they also significantly broaden the \emph{attack surface}. As a result, there has been a rise in cyberattacks targeting cyber-physical systems (CPS)~\cite{humayed2017cyber}.
\\
\noindent $\bullet$ \textit{\textbf{Timing Side-Channel and Schedule Based Attacks: }}To meet the stringent timing requirements and ensure reliability, the system designers take care to ensure
that the tasks within the application software execute in a \emph{deterministic} manner, i.e. to exhibit well-defined timing characteristics such as fixed activation intervals, bounded execution times, and minimum jitter, ensuring predictability and schedulability during run-time~\cite{huang2021detection,gandolfi2001electromagnetic}. Typically, a preemptive fixed-priority (PFP) scheduler is used to achieve static and deterministic task scheduling.
However, this inherent deterministic property of the system makes it vulnerable to \emph{timing-side channels}. By exploiting the deterministic nature of task schedules, an adversary can gain access to critical system information during run-time, such as temperature, safety-critical task run-time information, electromagnetic spikes, etc.~\cite {gandolfi2001electromagnetic}. In particular, the authors in ~\cite{chen2019novelsidechannel} demonstrated that by gaining access to a vulnerable task of the software, typically a \emph{untrusted, lower priority task}, an adversary can predict the future arrival instances of a safety-critical task. Since a static fixed-priority schedule is followed, the job execution sequence repeats every hyperperiod, which helps the attacker infer the initial offsets of the safety-critical task and
predict their future arrival times. The attacker can manipulate data inside the victim task’s shared I/O device buffer within a time window around those inferred future arrival time instances~\cite{chen2019novelsidechannel,schedguard++}. These vulnerable intervals, known as the Attack Effective Window (AEW), depend on system implementation details and can be experimentally characterized by the attacker~\cite{schedguard++}.
\\
\noindent $\bullet$ \textit{\textbf{Related Works: }}Several approaches have been proposed in the literature to prevent and mitigate SBAs. Broadly, these efforts can be categorised into the following three areas:
(i)\emph{ Schedule Randomization :} As the name implies, this technique randomizes the task execution schedule to mask predictable periodic patterns. The process involves dynamically computing the available slack time for jobs in the \emph{ready queue} and introducing controlled variations in their release or execution times, thereby reducing predictability and making it harder for an attacker to infer the system’s schedule~\cite{yoon2016taskshuffler,kruger2021randomization,kruger2018vulnerability}. However, as shown by the work~\cite{nasri2019pitfalls, sain2023work, sain2025maars}, \emph{attack-unaware randomization} may increase the vulnerability of safety-critical tasks and lead to a higher number of context switches, potentially degrading overall system performance. (ii) \emph{Temporal Isolation:} This technique enforces strict separation of CPU cycles between safety-critical and untrusted tasks. In~\cite{schedguard++}, untrusted tasks are blocked after the victim task completes. However, for larger task sets, this can lead to underutilization of CPU cycles.
(iii) \emph{Differential Privacy Techniques: } In this line of work, researchers have proposed leveraging differential privacy to obfuscate real-time task schedules and thereby mitigate scheduler side-channel attacks. A notable approach is the injection of Laplace-distributed  time delays into the scheduling 
process, as introduced by~\cite{chen2021indistinguishability}. 
However, while noise injection following a Laplace distribution can obscure a schedule to a potential adversary, the added randomness 
can increase the risk of deadline misses in hard real-time and safety-critical systems where strict timing guarantees are important. 
\par \noindent $\bullet$ \textit{\textbf{Novelty and Contributions: }} Current state-of-the-art research works have exclusively focused on schedule modification for reducing side-channels which reveal schedule information. However, none of them take into consideration the effect of such mitigation strategies on safety-critical control applications/tasks. Such tasks influence the physical dynamics of systems under control, and security-aware modification in their scheduling can have adverse effects on the performance of closed-loop systems.
To address these limitations, we propose a novel framework, \emph{Controlled Delay-Aware Secure Scheduling (SecureRT)}. The framework identifies the victim safety-critical control task in the system, and provides suitable \emph{job-level delays} to victim task instances to ensure security against timing-side channel attacks. As discussed earlier, the success of an SBA relies heavily on whether the untrusted tasks can execute within the AEW after the completion of the victim tasks. Therefore, it is important to determine an optimal job-level delay sequence that preserves the \emph{schedulability} of the task set, maintains the \emph{performance} of the victim control task under attack and guarantees ~\emph{security} against SBAs by minimizing the temporal overlap between the victim task's AEW and all untrusted task's executions.  
By reducing this temporal overlap and compensating for job-level delays, we preserve the performance of a real-time control task while simultaneously reducing its attack surface.
\begin{figure}[!htbp]
    \centering
    \vspace{-3mm}
    \includegraphics[width=0.48\textwidth,clip]{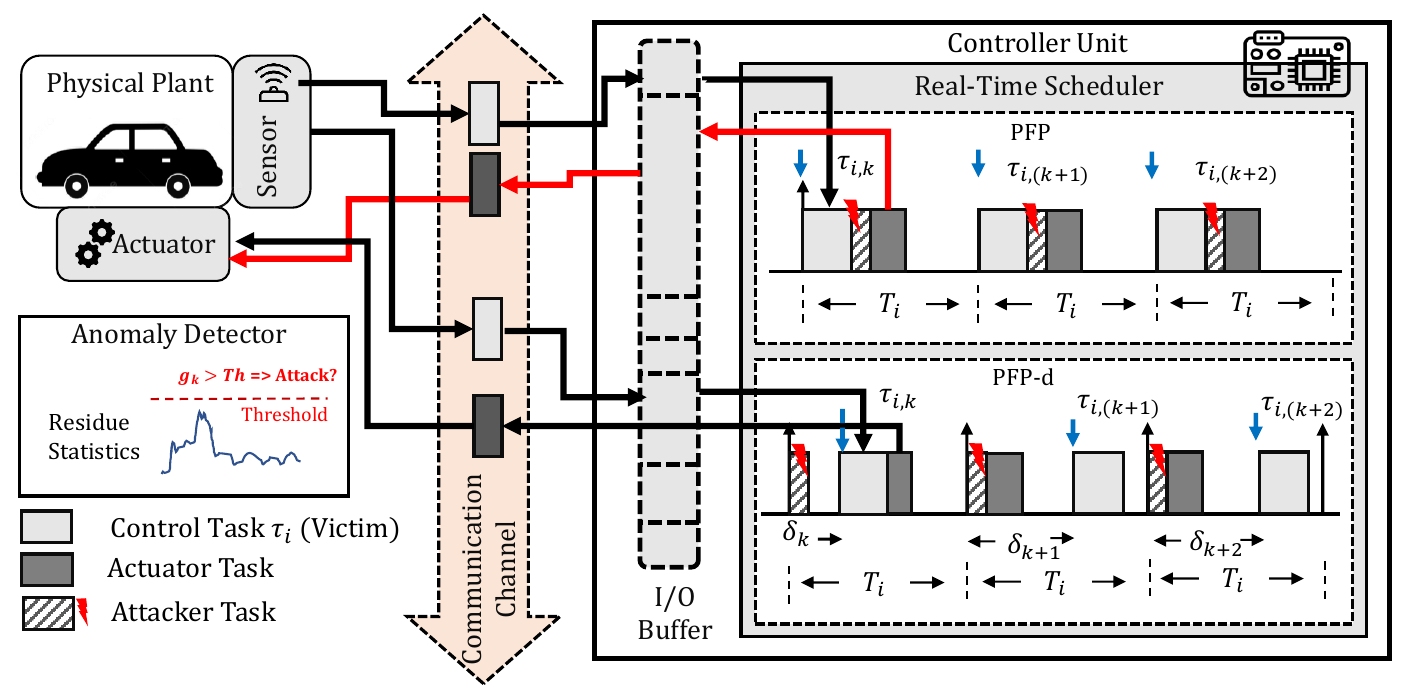}
    \caption{Overview of the System Model}
    \vspace{-2mm}
    \label{fig:sysmodel}
\end{figure}
We now summarize the main contributions of our work as follows.\\
\textbf{\textit{(1)} }We compute an upper bound on the job-release delay by performing a conservative worst-case response time (WCRT) analysis for the delayed victim task instances. We further derive a maximum admissible delay for which a delay-aware controller can still be synthesized with desired performance criteria.\\
\textbf{\textit{(2)} }We formulate an optimization problem to generate optimal job-level delay sequences for all control tasks, which minimizes the temporal overlap between AEWs of victim control tasks and all untrusted tasks.\\
\textbf{\textit{(3)} }We propose \emph{SecureRT}, a novel framework that integrates an online detection mechanism with a customized scheduler, termed as preemptive-fixed priority scheduler with delay (\emph{PFP-d}). When an attack is detected for a particular victim task, \emph{PFP-d} modifies its job release pattern by deploying the optimal job-level delay sequence to mitigate SBAs.\\
\textbf{\textit{(4)} }We evaluate the \emph{SecureRT} framework on a custom simulator that executes tasks on top of a real-time linux environment. Our experimental results demonstrate the trade-off between security and control performance under SBAs.
\section{System Model}
\label{sec:sysmodel}
\noindent $\bullet$ \textit{\textbf{Task and Scheduler Model}: } 
We consider $N$ independent periodic real-time tasks denoted by the set $\mathcal{T}= \{\tau_1, \tau_2,..., \tau_N\}$ on a single core platform scheduled by PFP scheduler. Each task $\tau_i \in \mathcal{T}$ is characterized by the tuple $(T_i, C_i, D_i)$ that states the minimum inter-arrival time (or period), worst-case execution time (WCET) and relative deadline. We assume that all tasks are independent and their deadlines do not exceed their respective periods, i.e., $D_i \leq T_i$. Given a task $\tau_i$, its corresponding higher and lower priority task sets are denoted using $hp(\tau_v)$ and $lp(\tau_v)$ respectively. We denote the set of control tasks from $\mathcal{T}$ by $\Gamma_C, \, \Gamma_C \subseteq \mathcal{T}$. The non-control tasks are denoted by the set $\Gamma_{NC} \, , \Gamma_{NC}= \mathcal{T} \setminus \Gamma_C$.
\par \noindent $\bullet$ \textit{\textbf{Control task Model}: }
Each control task $\tau_i \in \Gamma_C$ corresponds to a physical plant, which is modelled as a \emph{Linear Time-Invariant (LTI)} system with the following state-space dynamics: 
\begin{align}
\label{eq:continuous_system_equation}
\dot{x}= A_c x(t) + B_cu(t), \ y(t)= Cx(t)+v(t)
\end{align}
Here, the vectors $x \in \mathbb{R}^n, \hat{x} \in \mathbb{R}^n, y \in \mathbb{R}^m$, and $u \in \mathbb{R}^p$ describe the plant state, the estimated plant state, output and the control input, respectively. The matrices $A_c$ and $B_c$ are continuous-time matrices that define the continuous-time state and input-to-state transition matrices, respectively. $C$ is the output transition matrix. Each control task samples the sensed plant measurement vector $y$ at each sampling iteration, using which future states of the plant are estimated. Based on the estimated plant state, a control task calculates the required control input $u$ for actuating the plant in the subsequent sampling period. Fig.~\ref{fig:sysmodel} illustrates the interaction between the physical plant and the controller unit (indicated by the black box). The data flow between controller tasks and sensor/actuator units is illustrated by black arrows, while the data flow corrupted by an attacker is denoted by red arrows
We assume that the plant is discretized with periodicity $T_i$ and the sampled measurement $y[k]$ is available at time $kT_i$. Here, $T_i$ is also the period of the control task.  Considering the total delay in control task release and execution, let the control update $u[k]$ be available at $kT_i + \delta$, $\delta \in [0, T_i)$. 
Hence the plant uses $u[k-1]$ during $[kT_i, kT_i + \delta)$ and the updated input $u[k]$ during $(kT_i + \delta, (k+1)T_i)$. Therefore, the discrete-time dynamics of the control task updates according to the following equation:
\begin{align}
\label{eq:delayed_system}
x[k+1] = \Phi x[k] + \Gamma_1(\delta)\,u[k-1] + \Gamma_0(\delta)\,u[k]
\end{align}
where $\Phi = e^{A_c T_i}$, $\Gamma_0(\delta) = \int_{0}^{T_i - \delta} e^{A_c s} B_c \, ds$, and $\Gamma_1(\delta) = \int_{T_i - \delta}^{T_i} e^{A_c s} B_c \, ds$.
For small actuation delays $(\delta < T_i)$, the exponential terms can be approximated through first-order linearization, i.e.
\begin{align}
\Phi=I+ A_cT_i, \
\Gamma_0(\delta) \approx (T_i-\delta)B_c, \
\Gamma_1(\delta) \approx \delta B_c.
\end{align}
Following~\cite{chakraborty2016automotive}, we construct an \emph{augmented state-space model} to capture the dependency $x[k+1]$ on both $u[k]$ and $u[k+1]$, we define the augmented state as $z[k] = [x[k],u[k-1]]^T$, which updates the states in Eq.~\ref{eq:delayed_system} as: 
$z[k+1] = \Phi_{aug}(\delta)\,z[k] + \Gamma_{aug}(\delta)\,u[k],
\ y= C_{aug}z[k]$
where, $\scriptstyle
\Phi_{aug}(\delta)=
\begin{bmatrix}\Phi & \Gamma_1(\delta)\\[-1pt]0 & 0\end{bmatrix},\;
\Gamma_{aug}(\delta)=
\begin{bmatrix}\Gamma_0(\delta)\\[-1pt]I\end{bmatrix},\;
C_{aug}=
\begin{bmatrix}C & 0\end{bmatrix}.
$
Therefore, using the augmented matrices, the state equations evolve as:
$z[k+1] = \Phi_{aug}(\delta)\, z[k] + \Gamma_{aug}(\delta)\, u[k], \ y[k] = C_{aug}\, z[k]$
Each control task employs a state estimator and a feedback controller. The state estimator is implemented using a Kalman filter, which reconstructs the system states from the available output measurements. The control input is then generated using a Linear Quadratic Regulator (LQR) based on the estimated state.
%
\begin{align}
\label{Eq:final_sys}
\hat{z}[k+1] = 
\Phi_{aug}(\delta)\, \hat{z}[k] 
+ \Gamma_{aug}(\delta)\, u[k] \\ \nonumber
+ L_{aug}\big(y[k] - C_{aug}\, \hat{z}[k]\big),
u[k] = -K\, \hat{z}[k].
\end{align}
where $L_{aug}=[L \quad 0]^T$ and $K_{aug}=[K \quad 0]$ are the Kalman gain matrix and optimal feedback gain that minimizes a finite-horizon quadratic cost function:
\begin{align}
\label{eq:lqr_cost}
J(\delta) = \sum_{k=0}^{N} \left( z[k]^T Q_{aug} z[k] + u[k]^T R u[k] \right)
\end{align}
with $Q_{aug} \succeq$ 0 and $R \succ 0$ denoting the state and control weight matrices, respectively. In later sections, we use the cost function $J(\delta)$ to establish a bound on the maximum delay $\delta \in [0, T_i)$ that guarantees the closed-loop cost remains below a predefined threshold $J^{Th}$ set by the system designer over a finite duration.
\subsection{Residue-based Anomaly Detector}
\label{sec:detector_model} 
Each of the safety-critical control loops corresponding to the control tasks is equipped with a residue-based detector that observes changes in the system residue (i.e., the difference between the sensed and estimated outputs) for anomaly detection. We consider an integrated implementation of controller and detector functionalities as part of the control task to avoid data manipulation. We use a $\chi^2$-based detector in this work. The $\chi^2$ detector utilises a normalized quadratic function of the residual to amplify and easily detect minute variations in system residue. For system residue $res[k]$ at $k^{th}$ sampling iteration, its chi-square measure is $w[k] = res[k]^T \Sigma_{res}^{-1} res[k]$ where $\Sigma_{res}$ is the variance of system residue. Considering the measurement noise to be a zero-mean Gaussian noise, $res[k] \sim \mathcal{N}(0,\Sigma_{res})\Rightarrow w[k]\sim \chi^2(m,2m)$, where $m$ is the number of output measurements, considered as the degree of freedom of the $\chi^2$ distribution. We employ a windowed chi-square detector that compares the average value of the chi-square statistic of system residue over a pre-defined time window ($N$), i.e., $g[k]= \scriptstyle {\frac{1}{N}\sum\limits_{i=0}^{N-1}w[k-i]}$ and compares it with a pre-defined threshold $Th$. The threshold is calculated to maintain a desired false alarm rate~\cite{koley2021catch}. The chi-squared detector raises the alarm, denoting an attack attempt on a certain closed loop when $g[k]> Th$ at any $k^{th}$ sampling period of that closed loop. Such light weight attack detection methods have been found to have a high true positive rate in prior works \cite{murguia2016cusum}.
\section{Attack Model}
\label{threatmodel}
In this section, we discuss in-depth how we model the schedule-based attacker and explain our assumptions about the attacker’s objectives and capabilities. An attacker may gain entry to the system by exploiting vulnerabilities present in the wireless communication protocols used within the system~\cite{koscher2010experimental, woo2014practical}. 
A remote FDI attack on a Tesla vehicle was demonstrated by Tencent~\cite{nie2017free} where they exploited Wi-Fi/cellular interfaces to compromise the Instrument Cluster, Information Display and gateway, by injecting malicious CAN bus messages. 
Similarly, in this work, we consider an FDI attack on the control tasks ($\tau_i \in \Gamma_C$) within embedded platforms through the untrusted non-control tasks. Usually, control tasks, being safety critical in nature, are implemented 
by {\em trusted} vendors and OEMs, after thorough security and functionality checks and use authentication primitives.
Whereas an untrusted task $\tau_i \in \Gamma_U$, $ \Gamma_U \subseteq \Gamma_{NC}$, is implemented by a third party and may contain vulnerable software components. 
The attacker exploits an untrusted task and launches FDI on a control task during their interaction (common buffers, I/O etc) within a fixed time window, which we define as the Attack Effective Window (AEW)~\cite{ren2023protection,schedguard++}.
\par\noindent $\bullet$ \textit{\textbf{AEW: }} 
An Attack Effective Window $\Omega_i$ denotes a  time interval following either the completion of or before the start of a victim control task ($\tau_i$) instance. During this interval, 
\begin{wrapfigure}{l}{0.56\columnwidth}
  \centering
  \vspace{-6pt}
  \includegraphics[width=\linewidth]{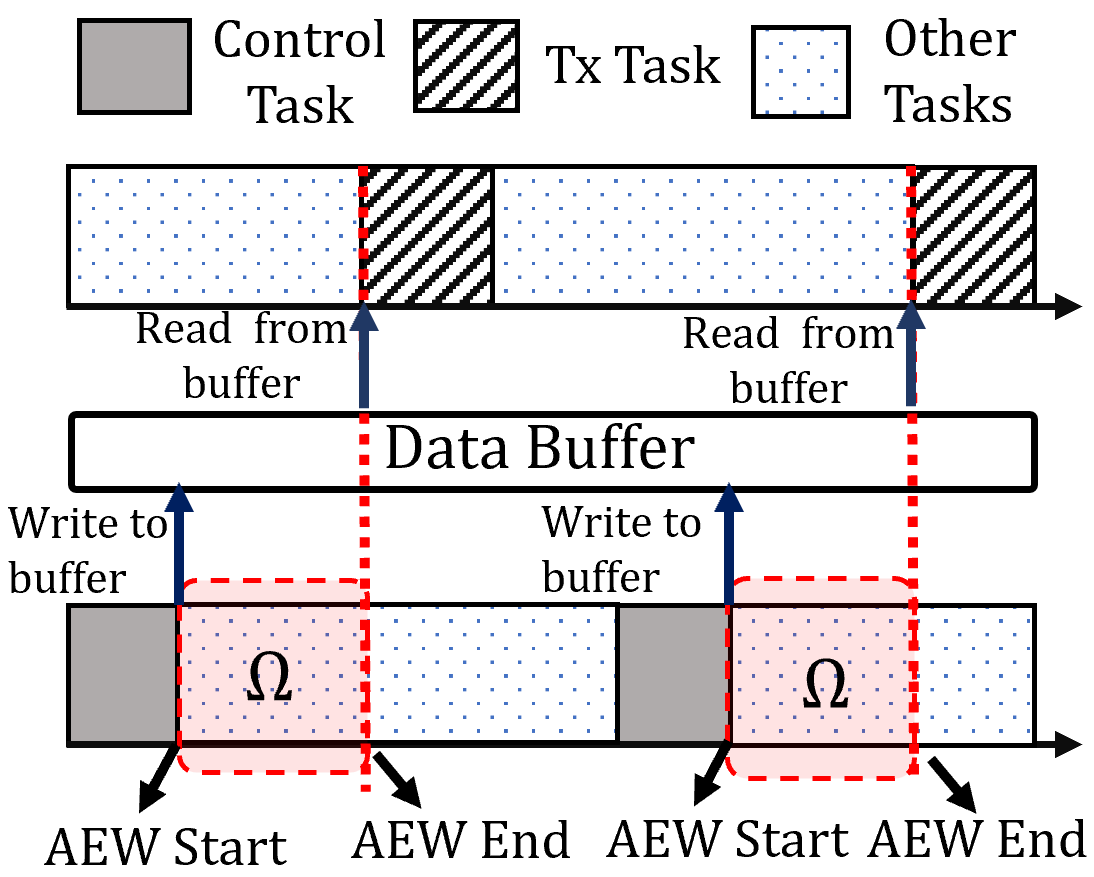}
  \caption{Attack Effective Window}
  \label{fig:aew}
  \vspace{-6pt}
\end{wrapfigure}
the attacker’s task execution can successfully modify the victim task’s data. The duration of this interval depends on the task schedule.     
In this work, we specifically focus on victim control tasks, which are typically assigned higher priorities in real-time CPSs. Therefore, lower-priority non-critical tasks often execute immediately after these control tasks within the AEW. Such cases, where the attacker executes after the victim to manipulate the control signal, are categorized as \emph{posterior attacks} in 
literature~\cite{nasri2019pitfalls,ren2023protection}. These attacks are particularly relevant for SBAs on control tasks, as they directly target the I/O or data buffer where control inputs calculated by the control task are stored after it finishes execution.
We illustrate this in Fig.~\ref{fig:aew}, where the AEW of the control task is shown shaded in red, considering a posterior attacker. AEW starts when the control task (dark shade) finishes execution, writes its data to the buffer and ends once the transmission task (striped box) starts. 
%
\par
$\bullet$ \textbf{\textit{Attacker Capabilities}: } 
 \textbf{\textit{(i)}} The task scheduler is assumed to be secure (similar to the works in~\cite{chen2019novelsidechannel,schedguard++}) and cannot be compromised by the attacker since the untrusted applications are forbidden from running with \emph{superuser/root privileges}. 
Hence, the adversary has no direct control over task execution and can only launch SBA through compromised tasks when they are scheduled by the system. 
\textbf{\textit{(ii)}} The attacker has full knowledge of the scheduling policy and can compromise a particular untrusted task $\tau_i \in \Gamma_U$. 
\textbf{\textit {(iii)}} The attacker can  partially estimate the sampling rates of any control task $\tau_i \in \Gamma_C$ by monitoring either the message data packets transmitted over the network or variations in the physical system states for several hyperperiods~\cite{chen2019novelsidechannel}.  The attacker untrusted task knows only its own start and finish times initially. However, based on preemption history, as demonstrated in the study of schedule-based attacks (SBAs) on a quadcopter, the AEW for a victim can be determined by the attacker ~\cite{chen2019novelsidechannel}.  \textbf{\textit {(iv)}} The attacker can modify the data produced by the victim control task within an I/O buffer or cache, only if the compromised task executes within the AEW~\cite{chen2019novelsidechannel,schedguard++}. 
\\
$\bullet$ \textbf{\textit{Attacker's Objective}: }We assume the attacker exploits timing information exposed by the task schedule (static PFP) to launch a posterior SBA on a safety-critical control task $\tau_v \in \Gamma_C$. By observing compromised lower-priority tasks $\tau_i \in \Gamma_U$, the attacker reconstructs the victim’s periodic job release pattern and runs the attacker task within the victim’s AEW to perform an FDI that degrades the performance of a closed-loop system. In the following section, we outline our mitigation strategy to counter such attacks. 
\section{Methodology}
In this section, we provide a detailed description of our proposed \emph{SecureRT} framework, which is applied to the victim control task under a schedule-based FDI attack. 
\begin{figure}[!htbp]
    \centering
    \vspace{-2mm}
    \includegraphics[width=0.5\textwidth]{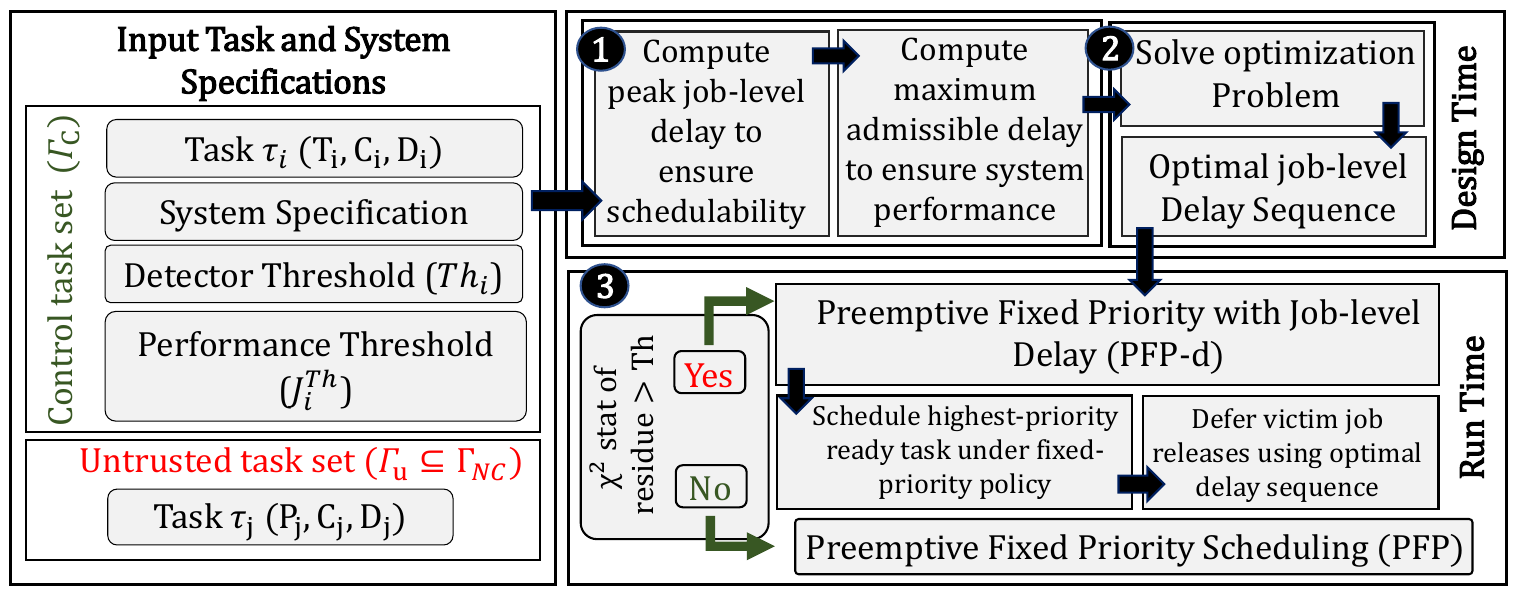}
    \caption{Overview of the \emph{SecureRT} framework}
    \vspace{-2mm}
    \label{fig:overview_methodology}
\end{figure}
A step-wise overview of the \emph{SecureRT} framework is shown in Fig.~\ref{fig:overview_methodology}, which operates through three sequential steps, which we will describe in brief. 
It takes the following inputs \emph{(i)} a task set $\mathcal{T}=\Gamma_C \cup \Gamma_{NC}$ consisting of both trusted control tasks ($\Gamma_C$) and untrusted task set ($\Gamma_U \subseteq \Gamma_{NC}$), respectively, characterized by period $T$, computation time $C$, deadline $D$, \emph{(ii)} physical system specifications corresponding to each control task $\tau_i \in \Gamma_C$, \emph{(iii)} Detector threshold $Th$ for all control tasks and \emph{(iii)} Control cost limit $J^{Th}$ set by system designer for all control tasks.
In \textbf{\em Step-1}, the framework determines the peak admissible job-level delay for each control task $\tau_i \in \Gamma_C$ using response-time analysis (RTA) to ensure that all schedulability constraints are satisfied. Following this, we compute the maximum admissible job-level delay that ensures that performance constraints are met for the corresponding control task.  
In \textbf{\em Step-2}, we formulate an MILP optimization problem to minimize the overlap between the AEW of the victim task and execution intervals of all untrusted tasks. This formulation incorporates the schedulability and performance constraints we described earlier. The solution of this MILP is an optimal set of job-delay sequences for all control tasks in the system that are stored offline (for later application). 
Finally, in \textbf{\em Step-3}, the job-level delay sequences stored offline are applied to the release times of the victim control task, and scheduling is done with \emph{PFP-d}. The resulting optimal task schedule eliminates predictable timing patterns, which the attacker might have used to launch an SBA. Moreover, it reduces the success rate of schedule-based FDI attack attempts in subsequent hyperperiods, while ensuring that control performance degradation remains within limit set by system designer. In the following sections, we discuss each of the system design steps and the run-time algorithm in detail.
\subsection{Peak Job-level Delay Derivation using RTA}
\label{sec:RTanalysis}
In RTS with PFP schedulers, determining whether each task can meet its respective deadline under worst-case conditions is important to ensure its timeliness. 
The \textit{Worst-Case Response Time (WCRT)} of a real-time task $\tau_i$ is computed for this purpose~\cite{liu1973scheduling}. It is the maximum time elapsed between the release of its job and its completion. Given any sporadic or periodic task $\tau_i$ scheduled on a uniprocessor under fixed-priority preemptive scheduling, the WCRT $R_i$ can be computed using the following recurrence relation:
\begin{equation}
\label{eq:wcrt}
R_i^{(n+1)} = C_i + \sum_{\tau_j \in hp(\tau_i)} \left\lceil \frac{R_i^{(n)}}{T_j} \right\rceil C_j
\end{equation}
Starting with $R_i^{(0)} = C_i$, the iteration proceeds until the response time converges i.e. $R_i^{(n+1)} = R_i^{(n)}$) or exceeds the task deadline ($R_i > D_i$), indicating that $\tau_i$ is unschedulable.
\par We first formally define the notion of \emph{job-level delays}, which represent temporal offsets applied to the release times of job instances. 
\begin{definition}[Job-Level Delay]
\label{def:job-level-delay}
Let $T_v$ denote the period of a safety-critical control task $\tau_v$, and let $H$ represent the hyperperiod of the task set, such that $N = H / T_v$ corresponds to the total number of releases of $\tau_v$ within one hyperperiod. The \textit{job-level delay sequence} of $\tau_v$ is defined as $\Delta_v[N] = \{\delta_1, \delta_2, \ldots, \delta_N\}$, 
where each $\delta_j$ represents a temporal offset applied to the $j$-th job instance of $\tau_v$. Accordingly, the modified release time of the $k$-th job is given by:
\begin{align}
    r'_{v,k} = r_{v,k} + \delta_j, \quad \forall k \ge 0,\ j = (k \bmod N) + 1,
\end{align}
where $r_{v,k}$ denotes the nominal release time, and the sequence $\Delta_v[N]$ is applied cyclically across job instances within each hyperperiod.
\end{definition}
We provide a strict upper bound on the delay $\delta_j<D_v-C_v, \ \forall j \leq N$, to prevent another job release with the current job pending. 
 In the following subsections, we will analyse the schedulability of both control and non-control tasks using Eq.~\ref{eq:wcrt}.
 Thereafter, we compute the WCRT of $\tau_i$, and establish a safe upper bound on its job-level delay that ensures $\tau_i$ continues to meet its timing guarantees, 
 even in the presence of perturbations to its job release pattern. Additionally, we calculate the WCRT of other non-victim tasks within the system.
\noindent \subsubsection{Response Time Analysis of Victim Task}
\label{sec:victim_schedulability}
\par In order to compute the worst-case response time (WCRT) of the $k$-th job of the victim task $\tau_v$, that has been delayed by $\delta$, we must account for all possible sources of interference that may delay its execution. Since there is no resource sharing between tasks, we can safely assume that tasks with priority lower than the victim control task $lp(\tau_v)$ will 
\begin{figure}[htbp]
\vspace{-5mm}
    \centering
\includegraphics[width=0.95\columnwidth,clip]{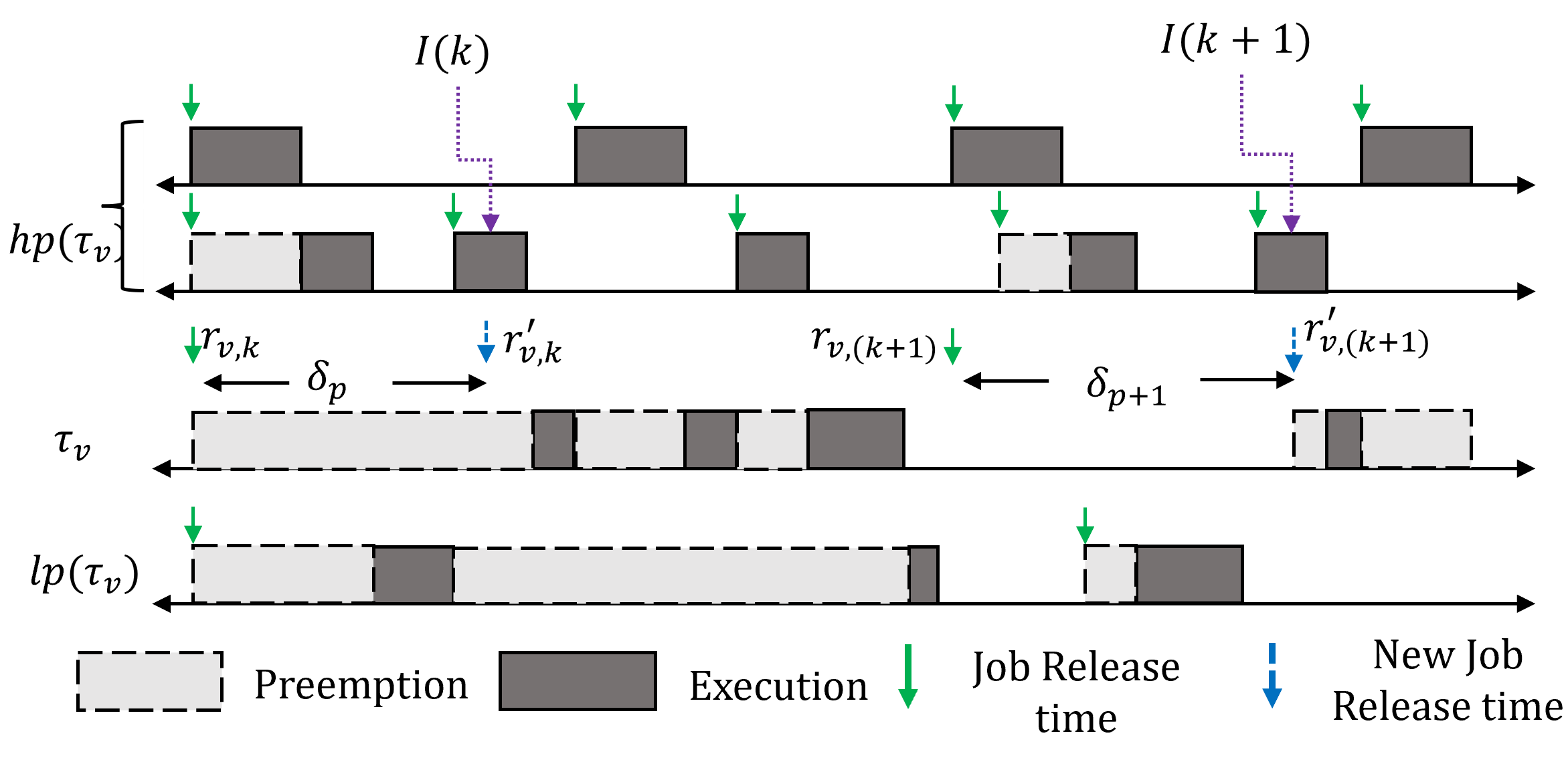}
    \caption{Interference faced by victim $\tau_v$ from $hp(\tau_v)$}
    \vspace{-3mm}
    \label{fig:interference}
\end{figure}
not block $\tau_v$. Note that jobs of $hp(\tau_v)$ that were released before $r'_{v,k}$ but have finished by $r'_{v,k}$ do not interfere with the victim's job released at time $r'_{v,k}$. Therefore, our analysis only involves identifying the \emph{carry-in} jobs from $hp(\tau_v)$ that were released \textit{before} the $r'_{v,k}$ and are still executing (when the victim's job releases after delay), as well as higher-priority jobs released \textit{after} $r_{v,k}$. We perform the RTA in two distinct steps as follows.
\textbf{(i)} We compute the interference $I(k)$ from jobs of $hp(\tau_v)$ that were released prior to the $k$-th instance of $\tau_v$, but remain unfinished at it's arrival. Such jobs are typically referred to as \emph{carry-in} jobs. \textbf{(ii)} We perform an overall WCRT analysis for the $k$-th instance of $\tau_v$. This involves computing the cumulative interference from all higher-priority tasks in $hp(\tau_v)$ that are released at or after the release time of the $k$-th job. The interference is accounted for during the entire \emph{response interval} of the job, defined as the time span between its release and completion.
\\
\noindent \textbf{(i) Computing Interference \(I(k)\):}
We denote the interference from \emph{carry-in jobs} for the $k$-th victim instance as $I(k)$, which means interference from jobs of higher-priority tasks that were released \textit{before} the release time $ r'_{v,k}$ of the victim job but are still \textit{executing at} time $r'_{v,k}$, i.e.  $r_{i,j} < r'_{v,k} \land r_{i,j} + C_i > r'_{v,k}$. 
\begin{lemma}
Given a uniprocessor system scheduled with PFP, let $\tau_v$ be a victim control task and $\tau_{v,k}$ denote its $k$-th job, released at time $r'_{v,k}$ with a delay $\delta$. The total interference $I(k)$ from the carry-in jobs of $hp(\tau_v)$ that were released before $r'_{v,k}$ and are still executing at $r_{v,k}$ is given by:
\begin{align}
\label{Eq:interference_S1}
I(k) = 
\sum_{\tau_i \in hp(\tau_v)} 
\max \!\scriptstyle{\left(
0, \ \Big\lceil \frac{r'_{v,k}}{T_j} \Big\rceil - \Big\lfloor \frac{r'_{v,k} - C_j}{T_j} \Big\rfloor-1\right)} C_j.
\end{align}
\end{lemma}
\begin{proof}
We aim to compute the cumulative interference from all higher-priority tasks $\tau_i \in hp(\tau_v)$ that have jobs released before $r'_{v,k}$ but are still executing at that time. For $j$-th carry-in job of task $\tau_i$ that is released at time $ j\cdot T_i$, we have $j \cdot T_i < r'_{v,k} \land j\cdot T_i + C_j > r'_{v,k}$. From these inequalities, we obtain bounds on the integer job index $j$. From \(j\cdot T_i + C_i > r'_{v,k}\) we have $\scriptstyle{j > \frac{r'_{v,k} - C_i}{T_i} \
\Rightarrow \ j_{\min}(i) = \Big\lfloor\frac{r'_{v,k}-C_i}{T_i}\Big\rfloor + 1}$. Similarly, from the inequality \(j\cdot T_i < r'_{v,k}\) we have $\scriptstyle{j < \frac{r'_{v,k}}{T_i} \
\Rightarrow \ j_{\max}(i) = \Big\lceil\frac{r'_{v,k}}{T_i}\Big\rceil - 1.}$
The integer indices $j$ that satisfy both these inequalities are those in the interval $[\,j_{\min}(i),\, j_{\max}(i)\,]$. Hence the number of carry-in jobs of task $\tau_i$ at time $r'_{v,k}$ is given by $j_{\max}(i) - j_{\min}(i) + 1$, which is:
$$\scriptstyle {\Big(\Big\lceil\frac{r'_{v,k}}{T_i}\Big\rceil - 1\Big)
- \Big(\Big\lfloor\frac{r'_{v,k}-C_i}{T_i}\Big\rfloor + 1\Big)
+ 1 \ \\
= \ \Big\lceil\frac{r'_{v,k}}{T_i}\Big\rceil
- \Big\lfloor\frac{r'_{v,k}-C_i}{T_i}\Big\rfloor
- 1}.$$
Since each carry-in job can contribute at most $C_i$ units of interference, we get 
$I(k,i) \ = \ \max \scriptstyle \!\left(0,\; 
\Big\lceil\frac{r'_{v,k}}{T_i}\Big\rceil
- \Big\lfloor\frac{r'_{v,k}-C_i}{T_i}\Big\rfloor
-1
\right) C_i.$
Summing over all $\tau_i \in hp(\tau_v)$, we get Eq.~\ref{Eq:interference_S1}.
\end{proof}
\noindent \textbf{(ii) Overall WCRT Analysis:}
The victim's job also experiences interference from jobs of $\tau_i \in hp(\tau_v)$ that were released at/after $r'_{v,k}$ within the response interval of the victim's job, i.e. $j \cdot T_i \ge r_{v,k} + \delta$. 
Accordingly, the WCRT of the $k$-th victim instance can be computed by solving the following fixed-point recurrence. 
\begin{align}
\label{eq:wcrt_recurrence}
R_{v,k}^{(n+1)}= C_v + I(k) +
\sum_{\tau_i \in hp(\tau_v)}
\Big\lceil\frac{R_{v,k}^{(n)}}{T_i}\Big\rceil \cdot C_i
\end{align}
We initialize $R_{v,k}^{(0)}=C_v$.
The iteration continues until convergence, i.e., when $R_{v,k}^{(n+1)} = R_{v,k}^{(n)}$.
 The effective deadline of the $k^{\text{th}}$ job, accounting for the release delay $\delta_k$, is given by $D_{v,k} = D_v - \delta_k$. For the victim task $\tau_v$, given the job level delay sequence $\Delta_v= \{\delta_1, \delta_2, \dots, \delta_{H/T_v}\}$, the schedulability condition for the victim task can be expressed as follows. 
\begin{align}
\label{eq:schedulability_victim}
    R_{v,k} \leq D_{v,k} \quad \forall k \in \{ 1, \dots, H/T_v \} 
\end{align}



Note that since the job-level delay sequence $\Delta_v$ is only applied to the jobs of $\tau_v$, it doesn't delay jobs of $hp(\tau_v)$, which can preempt jobs of $\tau_v$. However, jobs of $lp(\tau_v)$ will be affected as the delayed execution of $\tau_v$ can modify their start times, which can increase/decrease their response times. Hence, in the next section, we will perform a WCRT analysis for $\tau_i \in lp(\tau_v)$.
\subsubsection{Response Time Analysis of lower priority Non-victim Tasks}
\label{section:schedulability_nv}
To compute a conservative bound of the WCRT for lower priority tasks, i.e $\tau_i\in lp(\tau_v)$, we extend the standard response-time analysis to account for jitter~\cite{redell2002calculating} in the release of a higher-priority victim task $\tau_v$. 
For each regular (i.e. non-victim) task $\tau_i$ with priority lower than the victim task $\tau_v$, i.e $\tau_i \in lp(\tau_v)$, the interference is computed as follows.

\begin{align}
\label{eq:scheduability_non_victim}
R_i^{(k+1)}= C_i + \sum\limits_{\tau_j \in hp(\tau_i) \setminus \tau_v} \left\lceil \scriptstyle \frac{R_i^{(k)}}{T_j} \right\rceil C_j 
+ \max\!\Big(0, \left\lceil \scriptstyle \frac{R_i^{(k)} - \delta}{T_v} \right\rceil \Big) C_v
\end{align}
where $\delta = min\{ \delta_1, \delta_2, \dots, \delta_{H/T_v}\}$.
Eq.~\ref{eq:scheduability_non_victim} accounts for the interference from all other non-victim higher-priority tasks $\tau_j \in hp(\tau_i), j \neq v$ (second term) as well as the 
  the interference from the victim task $\tau_v$ with modified release times (third term). For the victim task, the response time interval is decreased in the worst case by the minimum induced delay $\delta$. 
   The recurrence is solved iteratively, similar to the standard fixed-priority Response Time Analysis (RTA) in Eq.~\ref{eq:wcrt} until convergence. 
    This formulation computes an upper bound on the response time of all lower-priority non-victim tasks, even under delayed releases of the higher-priority victim task $\tau_v$ instances. Given a job level delay sequence $\Delta_v= \{\delta_1, \delta_2, \dots, \delta_{H/T_v}\}$, of the victim task $\tau$, any task $\tau_i \in lp(\tau_v)$ is considered schedulable if, $R_i \le D_i$ where $R_i$ is computed using Eq. \ref{eq:scheduability_non_victim}. 
\subsubsection{Peak Job-level Delay and Overall System Schedulability}
\label{sec:peak-delay}
Given a job-level delay sequence $\Delta_v$ applied to job-release times $\tau_v$, to ensure overall system schedulability, the following criteria should be met. \emph{(i) }The victim task should meet its respective deadline under job-level delays as given by Eq.~\ref{eq:schedulability_victim} and \emph{(ii)} Lower priority non-victim tasks $lp(\tau_v)$ should meet their deadlines in the worst case as computed in Eq.~\ref{eq:scheduability_non_victim}. 
To determine a maximum bound of delay $\delta$ such that it preserves the schedulability of both the victim and all lower-priority tasks under a fixed-priority schedule, we define the \emph{peak job-level delay} of the victim task, denoted by $\Delta^{peak}_v$. It represents the largest delay $\delta \in \Delta_v$ such that it can be applied to all jobs of $\tau_v$ without violating the schedulability conditions of the system. Given any task set, if $\Delta_v^{peak}$ exists $\forall \ \tau_v \in \Gamma_C$, we consider the task set to be schedulable by \emph{SecureRT}.  Mathematically, we can write:
\begin{align}
\label{eq:peakdelay}
\Delta^{peak}_v &=
\max_{\delta \in [0,\, (T_v - C_v)]} 
\Big\{ \ \delta \;\Big|\;
 \forall k \in \{1,..,H/T_v \} \, R_{v,k} \leq D_{v,k}  \nonumber\\
&\hspace{6.5em} 
\bigwedge \forall\, \tau_i \in lp(\tau_v) \, R_i \leq D_i, \Big\}
\end{align}
where $R_{v,k}$ and $R_i$ represent the worst-case response times of $k$-th instance of $\tau_v$ and $\tau_i \in lp(\tau_v)$ respectively, under a job-level delay $\delta_k$. Given a single job-level delay $\delta$ applied to all the job samples of $\tau_v$, we denote its corresponding WCRT by $R_v(\delta)$ and for other lower priority non-victim tasks, it can be denoted by $R_i(\delta)$. In the next section, we discuss how to design a controller that meets our desired performance criteria under a job-level delay application.
\subsubsection{Maximum Delay Bound for Control Tasks}
\label{sec:control_job_delay}
Under a job-level delay $\delta$ applied to job instances of any victim control task $\tau_i$, the control signal actuation gets delayed by its corresponding WCRT $R_i(\delta)$, as derived in Eq.~\ref{eq:schedulability_victim}. As discussed in Sec.~\ref{sec:sysmodel}, the actuator holds the previously applied control input until the newly computed control update is actuated after $R_i(\delta)$. This degrades the closed-loop control performance. 
Therefore, this additional delay, $R_i(\delta)$, is used to model the delay-aware controller. 
To perform this, for any candidate job-level delay $\delta \in [0,\,\Delta_i^{peak}]$ (delay bound to ensure schedulability), we synthesise a corresponding \emph{delay-aware} controller with $R_i(\delta)$ following the methodology in Sec.~\ref {sec:sysmodel}.  
Specifically, we compute the delay-aware feedback gain $K_{aug}(R_i(\delta))$ and the Kalman estimator gain $L_{aug}(R_i(\delta))$ for the augmented discrete system in Eq.~\ref{Eq:final_sys} that models the actuation delay $R_i(\delta)$. These optimal controller and estimation gains are designed to optimize the cost of the delay-aware closed-loop system  $J_i(R_i(\delta))$ as computed in Eq.~\ref{eq:lqr_cost}.
\par While designing such an optimal controller, the system designer determines an upper bound, $J^{Th}$, of the control cost. This helps limit the performance degradation in the real-time platform due to platform-level uncertainties, such as delay and jitter.
We utilise this concept to determine a {\em maximum admissible delay} that preserves the control performance by limiting the control cost degradation for each control task $\tau_i \in \Gamma_C$ within a threshold $J_i^{Th}$. This maximum admissible delay must also ensure schedulability of the tasks scheduled in a uniprocessor, i.e., $\delta \in [0,\,\Delta_i^{peak}]$. 
We define the maximum admissible delay 
as follows:
\begin{definition}[Maximum Admissible Delay]
\label{def:max_delay}
The \emph{maximum admissible delay} of a control task $\tau_i$, denoted by $\Delta_i^{max} \in [0,\Delta_i^{peak}]$, 
is defined as the largest job-level delay that can be applied to all job instances of $\tau_i$ such that the nominal control cost computed using the corresponding delay-aware optimal control gain remains below a predefined performance threshold $J_i^{Th}$, 
i.e., $J_i(R_i(\delta)) \le J_i^{Th}$ and  $\Delta_i^{\text{peak}}$ satisfies Eq.~\ref{eq:peakdelay}. Mathematically, we can write:
\begin{align}
\label{eq:maxdelay}
\Delta_i^{max} = \max_{\delta \in [0,\, \Delta_i^{peak}]} \big\{\, \delta \;\big|\; J_i(R_i(\delta)) \le J_i^{Th} \,\big\}
\end{align}
\end{definition}
$\Delta_i^{max}$ represents the maximum delay that can be applied to the jobs of $\tau_i$ while ensuring both schedulability and desired control performance. Therefore, any delay $\delta_k \le \Delta_i^{max}$ can be applied to the $k$-th sample of $\tau_i$, without hampering the control performance, under the designed delay-aware controller. We demonstrate this with example~\ref{ex:example_1}.
\begin{example}
\label{ex:example_1}
We consider a task set $\tau = \{\tau_1,\tau_2,\tau_3,\tau_4\}$ scheduled with PFP, where $\tau_2$ is the victim task. The task parameters are shown in Tab.~\ref{tab:example_taskset2}. Our objective is to find the peak job-level delay $\Delta_2^{peak}$ (to ensure schedulability) and maximum admissible job-level delay $\Delta_2^{max}$ (to ensure control performance). All time units are in $ms$.
\begin{table}[H]
\centering
\begin{tabular}{|c||c|c|c|c|c|}
\hline
Task & $C_i$ & $T_i$ & $D_i$ & Vulnerability & $\Delta_i^{peak}$ \\
 & $[ms]$ & $[ms]$ & [ms] & Category & [ms]\\
\hline
\hline
$\tau_1$ & 1 & 5  & 5  & Trusted & -\\
\hline
$\tau_2$ & 3 & 10  & 10  & Trusted (Victim) & 6\\
\hline
$\tau_3$ & 3 & 20 & 20 & Untrusted & -\\
\hline
$\tau_4$ & 2 & 20 & 20 & Trusted & -\\
\hline
\end{tabular}
\caption{Example Task Set for $\Delta_i^{peak}$ Computation.}
\vspace{-2mm}
\label{tab:example_taskset2}
\end{table}
\begin{figure}[htbp]
    \centering
    \vspace{-2mm}    \includegraphics[width=0.5\textwidth,clip]{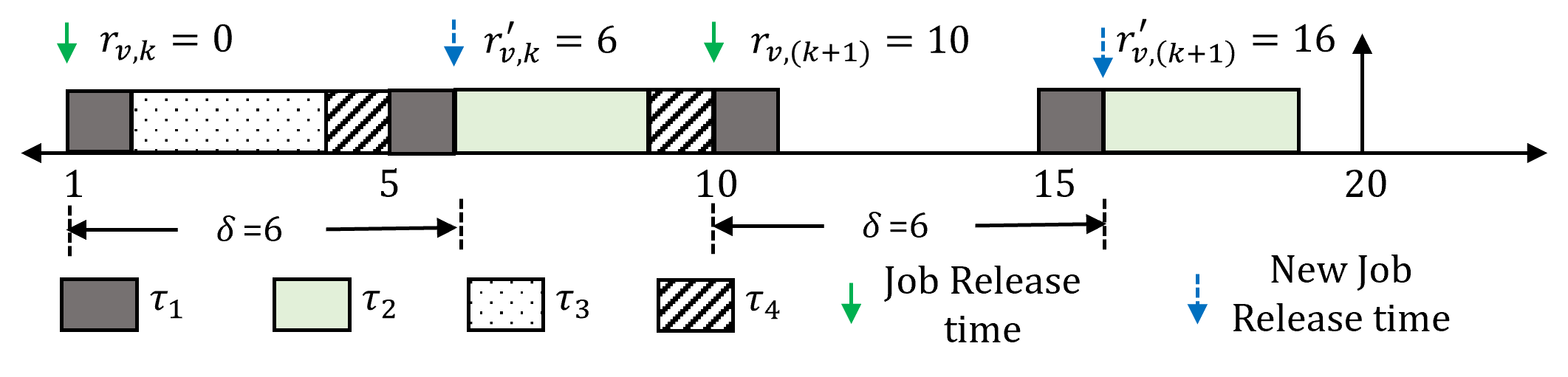}
    \caption{Task schdule with job-level delay $\delta =6$}
    \vspace{-6mm}
    \label{fig:example_schedule}
\end{figure}
If both samples of $\tau_2$ within the hyperperiod length $H=20$ are given $\delta=6$, the modified release times are $r'_{2,1}=6$ and $r'_{2,2}=16$. Using the interference computation for carry-in jobs from Eq.~\ref{Eq:interference_S1}, $I(0)=0, I(1)=0 $, indicating there are no carry-in jobs. We computed $R_{2,1}=4, R_{2,2}=4$ using Eq.~\ref{eq:wcrt_recurrence}. Since, $D_{2,1}=D_{2,2}=10-6=4$, both the jobs of $\tau_2$ meet their deadlines.
For the lower-priority tasks, $lp(\tau_2)=\{\tau_3, \tau_4$\}, using  $\delta=6$ in Eq.~\ref{eq:scheduability_non_victim}, the WCRT converges at values of $R_3(6)=4$ and $R_4(6)=10$, both meeting their respective deadlines $D_3=20$ and $D_4=20$. The higher priority tasks $hp(\tau_2)=\{\tau_1\}$ are not affected by delay. For $\delta > 6$, the jobs of $\tau_2$ miss their deadlines. Hence, all tasks remain schedulable for $\delta=\Delta_2^{peak}=6$, which represents the peak job-level delay that preserves overall system schedulability. The resulting task schedule has been illustrated in Fig.~\ref{fig:example_schedule}.
\par Using a trajectory-tracking controller (TTC) corresponding to task $\tau_2$, with a sampling time of $T_2 = 10$, we synthesised an optimal delay-aware controller as discussed in Sec.~\ref {sec:sysmodel}. We computed the control cost for 30 samples, corresponding to various job-level delays $\delta \in [0, \Delta_2^{peak}]$, using Eq.~\ref{eq:lqr_cost}. The plots of control cost are illustrated in Fig.~\ref {fig:example_cost}. The cost threshold $J_2^{Th}=139.3$ (in red dashed line) is set at $5\%$ higher than the cost at $\delta=0$ (in black dotted line). When a job-level delay $\delta=6$ has been applied, the cost $J_2(R_2(6))=142.6$ (in blue line) crosses the threshold $J_2^{Th}$ (in red dashed line), which means it is not possible to synthesise an optimal delay-aware controller. We simulated for all candidate delays $[0, 6]$ and observed that at $\delta=3\,$, $J_2(R_2(3))=138.2$, $J_2(R_2(3))<J^{Th}$. Therefore, $\Delta_2^{max}=3$, which suggests that the maximum admissible delay is $3ms$ that can be applied to the jobs of $\tau_2$, while ensuring both schedulability and desired control performance.
\end{example}
\subsection{Formulation of the Optimization Problem}
\label{sec:prob_formulation}  
As established in the earlier sections, we determine the maximum admissible job-level delay $\Delta_i^{max}, \ \forall \tau_i \in \Gamma_C$, such that any delay $\delta_k \leq \Delta_i^{max}$ preserves both system schedulability and the control performance of $\tau_i$. In this section, our goal is to find an optimal job-level delay sequence $\Delta_i[N_i], \ N_i = H/T_i$. The idea is to shift the release of every job instance of $\tau_i$ (by using the optimal job-level delay sequence) that ensures \emph{temporal overlap} between the duration of AEW $\Omega_i$ of the victim control task $\tau_i$ and execution duration of job instances of $\tau_j \in \Gamma_U$ is minimized. We denote the $k$-th job instance of victim task $\tau_i \in \Gamma_C$ by $\tau_{i,k}$, and $m$-th job instance of an untrusted task $\tau_j \in \Gamma_U$ by $\tau_{j,m}$.
In Fig.~\ref{fig:overlap}, we illustrated a case where the untrusted task instance $\tau_{j,m}$ (striped box) arrives within the AEW (red shaded box) after victim task instance $\tau_{i,k}$ (dark shaded box) has finished executing.
Note that $\tau_{i,k}$ under job-level delay $\delta_k$ finishes at time $f'_{i,k}$, where $f'_{i,k} = r_{i,k} +\delta_k + R_i(\delta_k)$ is the worst case finishing time.  
\par If the release time of victim task instance $\tau_{i,k}$ is delayed by $\delta_k$, an $m$-th instance ($m\in [1,H/T_j]$) of an untrusted task $\tau_j$ executes within $[r_{j,m}, r_{j,m} + R_j(\delta_k)]$ duration, where, $R_j(\delta_k)$ is the WCRT of $\tau_{j,m}$.
\begin{figure}[h]
    \centering
    \vspace{-2mm}
    \includegraphics[width=0.40\textwidth]{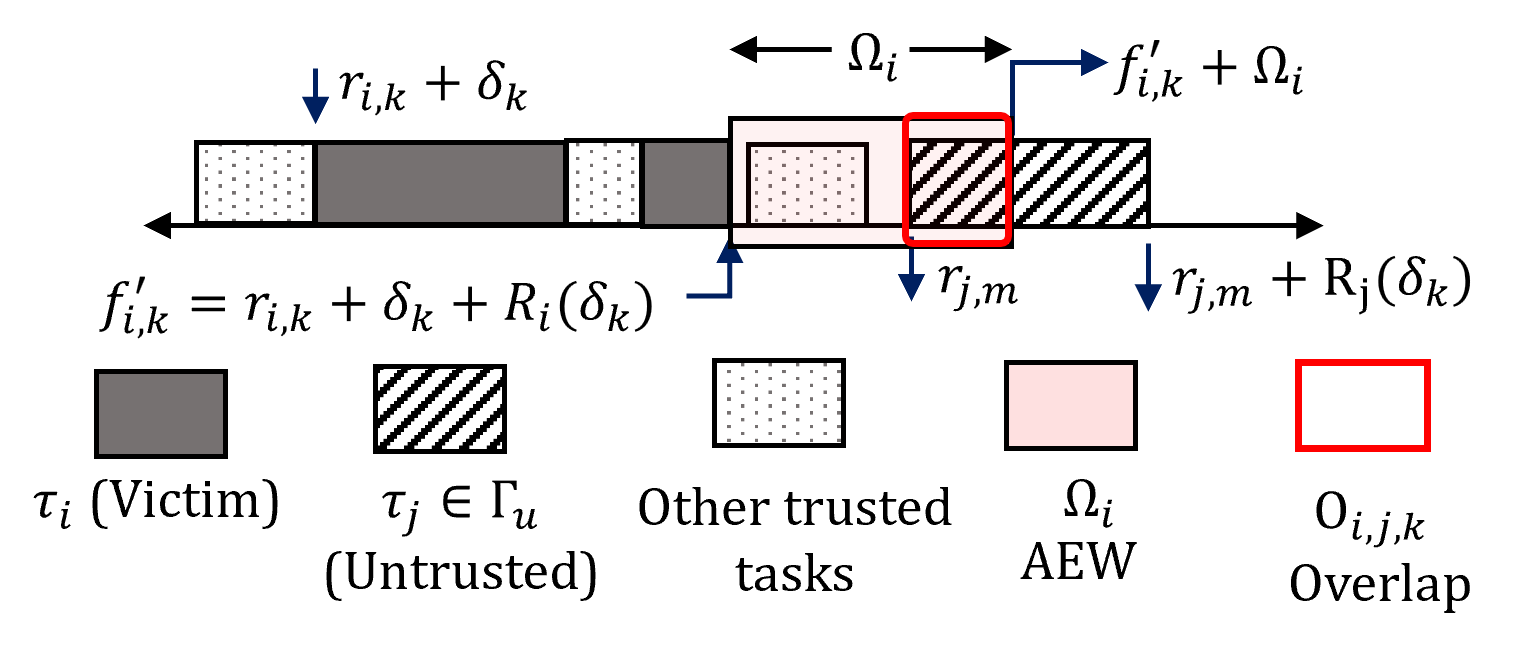}
    \caption{Overlap between victim's AEW and Untrusted Tasks }
    \label{fig:overlap}
    \vspace{-2mm}
\end{figure}
The AEW of the victim task instance $\tau_{i,k}$ is  $[f'_{i,k}, f'_{i,k} + \Omega_i]$.
The overlap $\mathcal{O}_{\langle i,k \rangle}^{\langle j,m \rangle}$ with this AEW of $\tau_{i,k}$ and untrusted task instance $\tau_{j,m}$'s execution time can be computed using the expression $\max\left(0, \min(f'_{i,k} + \Omega_i, r_{j,m} + R_j(\delta_k)) - \max(f'_{i,k}, r_{j,m})\right)$. 
Note that $\min(f_{i,k} + \Omega_i, r_{j,m} + R_j(\delta_k))$ determines the earliest finishing time between $\tau_{i,k}$'s AEW and $\tau_{j,m}$'s execution window. Also, $\max(f_{i,k}, r_{j,m})$ 
\begin{wrapfigure}[9]{l}{0.63\columnwidth}
  \centering
  \vspace{-6pt}
  \includegraphics[trim={0.1cm 0.1cm 0.2cm 0.1cm},width=\linewidth,clip]{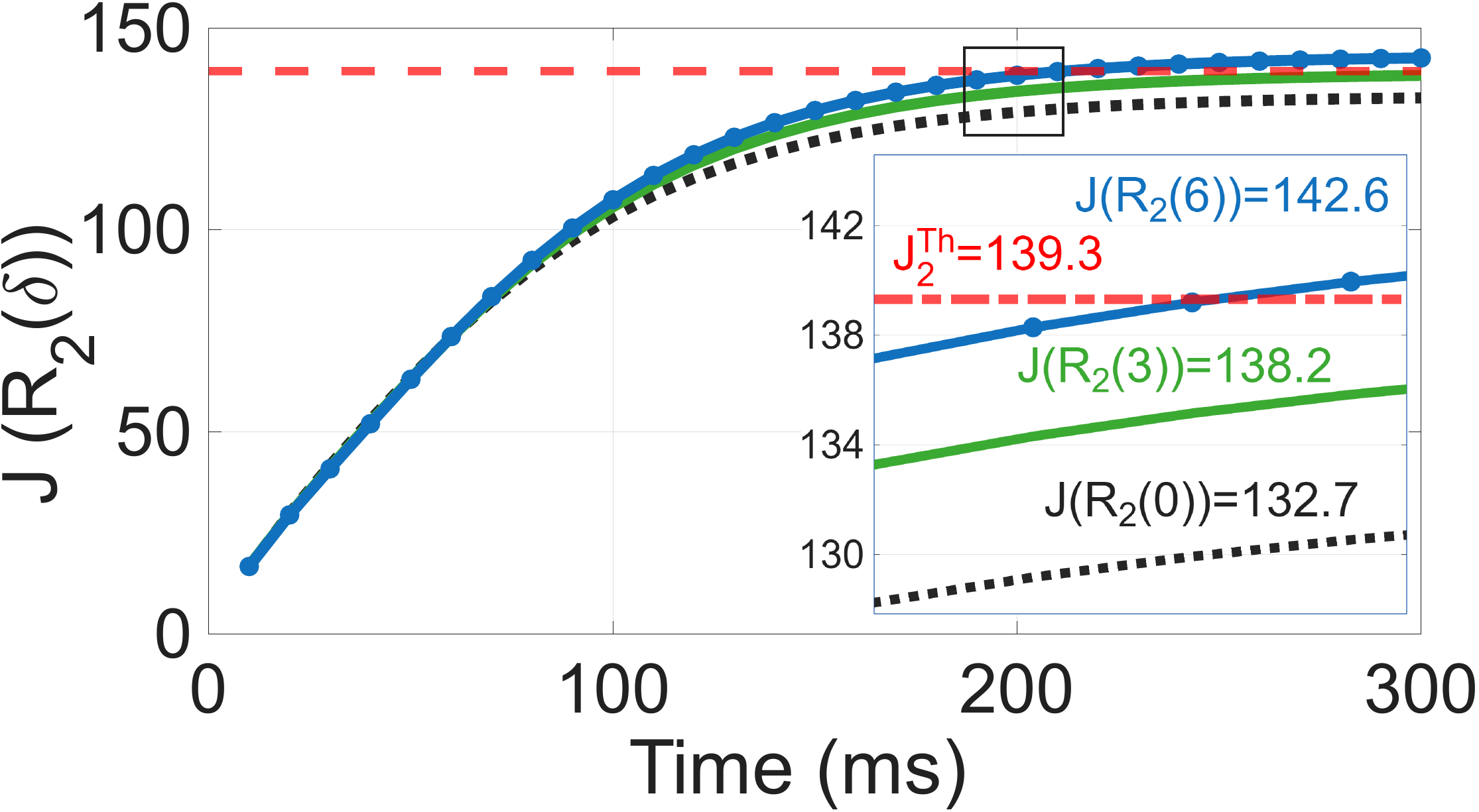}
  \caption{Determining $\Delta_i^{max}$ for $\tau_2$}
  \vspace{-2mm}
  \label{fig:example_cost}
\end{wrapfigure}
computes the latest start time between $\tau_{j,m}$ and $f'_{i,k}$. Finally, the outermost $max(0,.)$ function ensures that the computed overlap is zero when the above intervals do not intersect.
We illustrated one case of overlap in Fig.~\ref{fig:overlap}. Here, the overlap $\mathcal{O}_{\langle i,k \rangle}^{\langle j,m \rangle}$ is shown by a red-outlined box. In this case, the overlap corresponds to the condition $\min(f_{i,k} + \Omega_i, r_{j,m} + R_j(\delta_k)) = f'_{i,k} + \Omega_i$, indicating that AEW of the $\tau_{i,k}$ ends earlier than execution window $\tau_{j,m}$.
Note that the response times vary as a function of the job-level delay $\delta_k$. Therefore, as $\delta_k \in [0,\Delta^{max}_i]$, both the victim and untrusted tasks may experience WCRTs which may range from $R_j(\delta_k) \in [C_j, R_j(\Delta_i^{max})]$ and $R_i(\delta_k) \in [C_i, R_i(\Delta^{max}_i)]$. By suitably replacing the non-linear delay-dependent WCRTs in the expression for $\mathcal{O}_{\langle i,k \rangle}^{\langle j,m \rangle}$, we can derive an upper bound that computes the worst-case overlap between $\tau_{i,k}$'s AEW and $\tau_{j,m}$'s execution window. The overlap can be bounded as $\max\left(0, \min(f'_{i,k} + \Omega_i, r_{j,m} + R_j(\delta_k)) - \max(f'_{i,k}, r_{j,m})\right)\le \max \Big(
  0,\,
  \min\!\big( r_{i,k} + \delta_k + R_i(\Delta_i^{max}) + \Omega_i,\,              r_{j,m} \nonumber+ R_j(\Delta_i^{max}) \big) -\, \max\!\big( r_{i,k} + C_i + \delta_k,\,r_{j,m} \big)\Big)$. Here, the $\min(\cdot)$ term is replaced with the upper bounds of the response times $R_i(\Delta_i^{max})$ and $R_j(\Delta_i^{max})$, whereas the $\max(\cdot)$ term uses their lower bounds, which gives us a conservative (pessimistic) estimation of the maximum possible overlap. We present this pessimistic upper bound of overlap with the following notation:
\begin{align}
\label{eq:overlap_pess}
\mathcal{O}_{\langle i,k \rangle}^{\langle j,m \rangle}
  &= \max \Big(
  0,\,
  \min\!\big( r_{i,k} + \delta_k + R_i(\Delta_i^{max}) + \Omega_i,\,
              r_{j,m} \nonumber\\ 
   & + R_j(\Delta_i^{max}) \big) -\, \max\!\big( r_{i,k} + C_i + \delta_k,\,r_{j,m} \big)\Big)
\end{align}
By summing $\mathcal{O}_{\langle i,k \rangle}^{\langle j,m \rangle}$ over all job instances $m \in [1, T_j]$ of each untrusted task $\tau_j \in \Gamma_U$, and across all victim job samples $k \in [1, N_i]$, we obtain the total overlap duration $\mathcal{O}_i$ accumulated over the entire hyperperiod $H$, i.e.
\begin{align}
\label{eq:overlap_funcnonlin}
\mathcal{O}_{i} 
  &= \sum_{k=1}^{N_i} \sum_{\tau_j \in \Gamma_U} \sum_{m=1}^{H/T_j} \mathcal{O}_{\langle i,k \rangle}^{\langle j,m \rangle}.
\end{align}
A larger overlap value implies that a greater number of untrusted task instances' execution intervals coincide with the victim’s AEW, increasing the exposure of the control task to potential FDI attempts via posterior SBAs. Therefore, our objective is to compute an optimal job-level delay sequence $\Delta_i[N_i] = \{d_1, \dots, d_{N_i}\}$ that minimizes this overlap while preserving the schedulability of all tasks.
\\
$\bullet$ \textbf{Linearizing the Overlap Function: }Note that each term $\mathcal{O}_{\langle i,k \rangle}^{\langle j,m \rangle}$ in Eq.~\ref{eq:overlap_pess} includes a nested $\max(\cdot)$ and $\min(\cdot)$ operation, hence every summand term of $\mathcal{O}_i$ from Eq.~\ref{eq:overlap_funcnonlin} is a \emph{piecewise-linear} and \emph{non-convex} function of the $\delta_k \ \forall k \in [1,N_i]$. Therefore, the function in Eq.~\ref{eq:overlap_funcnonlin} can not be handled by standard MILP solvers~\cite{boyd2004convex}. So, we linearize it by introducing \emph{auxiliary} and \emph{binary decision variables} that transforms such non-linear formulation into an equivalent set of linear constraints in the form of $Ax\leq b$, which are suitable for solving MILP optimization. 
Let $M>0$ denote a sufficiently large constant (commonly referred to as the \emph{Big-M}) that upper bounds all relevant job release/finish times within the system. 
In practice, we choose $M \ge H + \max_{\tau_j \in \Gamma_U} R_j(\Delta_i^{max}) + R_i(\Delta_i^{max}) + \Omega_i$, which safely bounds all temporal variables appearing in the overlap computation. This ensures that all linearized inequalities remain valid for any feasible scheduling scenario. We define three continuous auxiliary variables $a_{k,j,m},\ b_{k,j,m},\ z_{k,j,m} \in [0,M)$ and three binary decision variables $y^a_{k,j,m},\, y^b_{k,j,m},\, o_{k,j,m} \in \{0,1\}$. The continuous variables indicate intermediate quantities from Eq.~\ref{eq:overlap_pess}. Specifically, $a_{k,j,m} = \max\!\big( r_{i,k} + C_i + \delta_k,\,r_{j,m} \big)$, $b_{k,j,m} = \min ( r_{i,k} + \delta_k + R_i(\Delta_i^{max}) + \Omega_i, r_{j,m} + R_j(\Delta_i^{max}))$ and $z_{k,j,m} = max(0,\, b_{k,j,m} - a_{k,j,m})$. 
The binary decision variables are used to indicate the active region of the piecewise-linear function. The variable $y^a_{k,j,m}$ determines whether $a_{k,j,m}$ equals $r_{i,k} + C_i + \delta_k$ or $r_{j,m}$, $y^b_{k,j,m}$ determines whether $b_{k,j,m}$ equals $r_{i,k} + \delta_k + R_i(\Delta_i^{max}) + \Omega_i$ or $r_{j,m} + R_j(\Delta_i^{max}))$ and $o_{k,j,m}$ indicates whether a positive overlap (\(z_{k,j,m}>0\)) exists. 
\par For the expression $a_{k,j,m} = \max\!\big( r_{i,k} + C_i + \delta_k,\,r_{j,m} \big)$, following are the linear constraints: 
\begin{align}
    a_{k,j,m} &\geq r_{i,k} + C_i + \delta_k \label{eq:a_ge_f}\\
    a_{k,j,m} &\geq r_{j,m} \label{eq:a_ge_r}\\
    a_{k,j,m} &\leq r_{i,k} + C_i + \delta_k + M(1 - y^a_{k,j,m}), \label{eq:a_le_f}\\ 
    a_{k,j,m} &\leq r_{j,m} + M\,y^a_{k,j,m} \label{eq:a_le_r}
\end{align}
If $y^a_{k,j,m}=1$, Eq.~\eqref{eq:a_ge_f}-~\eqref{eq:a_le_f} gives $a_{k,j,m}= r_{i,k} + C_i + \delta_k$ and if $y^a_{k,j,m}=0$, Eq.~\eqref{eq:a_ge_r}–~\eqref{eq:a_le_r} gives $a_{k,j,m}= r_{j,m}$. 
%
Similarly, for the expression $b_{k,j,m} = \min ( r_{i,k} + \delta_k + R_i(\Delta_i^{\max}) + \Omega_i, r_{j,m} + R_j(\Delta_i^{\max}))$, we introduce the linearized constraints below:
\begin{align}
    b_{k,j,m} &\leq r_{i,k} + \delta_k + R_i(\Delta_i^{\max}) + \Omega_i, \label{eq:b_le_f}\\
    b_{k,j,m} &\leq r_{j,m} + R_j(\Delta_i^{\max}), \label{eq:b_le_r}\\
    b_{k,j,m} &\geq r_{i,k} + \delta_k + R_i\scriptstyle{(\Delta_i^{\max})} + \Omega_i - M(1 - y^b_{k,j,m}), \label{eq:b_ge_f}\\
    b_{k,j,m} &\geq r_{j,m} + R_j\scriptstyle{(\Delta_i^{\max})} -  M\,\scriptstyle{y^b_{k,j,m}}. \label{eq:b_ge_r}
\end{align}
Finally, for the outermost $max(.)$ function \(z_{k,j,m} = \max(0,\, b_{k,j,m} - a_{k,j,m})\) the linearized constraints are:
\begin{align}
    z_{k,j,m} &\ge 0, \label{eq:z_nonneg}\\
    z_{k,j,m} &\ge b_{k,j,m} - a_{k,j,m}, \label{eq:z_ge_diff}\\
    z_{k,j,m} &\le b_{k,j,m} - a_{k,j,m} + M(1 - o_{k,j,m}), \label{eq:z_le_diff}\\
    z_{k,j,m} &\le M\,o_{k,j,m}. \label{eq:z_le_bigM}
\end{align}
Therefore, the complete MILP formulation for optimal job-level delay synthesis can be specified as:
\begin{gather}
\min_{\substack{
a_{k,j,m},\, b_{k,j,m},\, z_{k,j,m},\\
y^a_{k,j,m},\, y^b_{k,j,m},\, o_{k,j,m},\,\delta_k
}}
\mathcal{O}_i = \sum_{k=1}^{H/T_i} 
\sum_{\tau_j \in \Gamma_U} 
\sum_{m=1}^{H/T_j} 
z_{k,j,m}, \label{eq:objective_func_lin}\notag
\end{gather}
\begin{align}
\text{subject to:} \quad &(1) \quad 0 \le \delta_k \le \Delta_i^{\max}, \forall k \in [1,N_i] \tag{C1} \label{constr:C1} \\
&(2) \quad \text{Eqs.}~(\ref{eq:a_ge_f})\text{--}(\ref{eq:z_le_bigM}). \tag{C2} \label{constr:C2}
\end{align}
Solving the MILP in Eq.~\ref{eq:objective_func_lin} provides a concrete set of optimal job-level delay $\Delta_i[N]= \{\delta_1,\delta_2...\delta_{H/T_i}\}$ values for the control task $\tau_i$. 
Eq.~\eqref{eq:a_ge_f}--\eqref{eq:z_le_bigM} collectively introduce 
$\left[6 \left( \tfrac{H}{T_i}\lvert \Gamma_U \rvert  \right) + \tfrac{H}{T_i}\right]$
decision variables, including both continuous and binary variables. 
The linearization constraints in~\ref{constr:C2} contribute 
$12\!\left( \tfrac{H}{T_i} \lvert \Gamma_U \rvert\right)$ linear constraints and from schedulability and performance constraints in~\ref{constr:C1}, we get $2(H/T_i)$ linear constraints. 

\subsection{Runtime Algorithm for SecureRT Framework}
\label{runtime_algorithm}
\begin{algorithm}[!b]
\footnotesize
\caption{\emph{PFP-d} Scheduler}
\label{alg:pfp_d}
\begin{algorithmic}[1]
\Require Task set $\mathcal{T}=(\Gamma_C\cup\Gamma_{NC}, \Gamma_U \subseteq \Gamma_{NC})$, job-level delays $\Delta_{i_{\tau_i\in\Gamma_C}}$, detector flag \textit{AtkFlag} (returns $v$ for victim $\tau_v$ or $-1$), thresholds $Th_i$, time step $\Delta t$, reset time $T_{RES}$
\Ensure Execution sequence $S$
\State $Mode \gets$ PFP, $t \gets 0$, $AtkFlag=-1$, $JobIdx \gets 0$ \label{init}
\State $\Delta_v \gets \text{NULL}$; $ReadyQueue, DeferQueue, S \gets \emptyset$ \label{init_arrays} 
\While{$t \le T_{RES}$} \label{while_loop}
  \State $(g[k],\, i) \gets \textsc{UpdateDetector}(t)$  \Comment{Check detector output}  \label{update_detector}
  \If{$(g_i[k] > Th_i)$ \textbf{and} $(i \neq -1)$}
     \State $AtkFlag \gets i$; $v \gets i$;
     \State $\Delta_v \gets \Delta_i$; \Comment{Fetch victim’s delay pattern}
     \State $JobIdx \gets k \bmod |\Delta_v|$; $Mode \gets$ PFP-d \label{switch_mode_kth}
  \EndIf
  \For{each $\tau_i \in \mathcal{T}$ arriving at time $t$} \label{for_arrival}
    \If{$(\tau_i = \tau_v)$ \textbf{and} $(Mode = \text{PFP-d})$}
        \State Defer $\tau_i$ to $t + \Delta_v[JobIdx] \cdot \Delta t$; \Comment{Apply delay} \label{defer_victim}
        \State $JobIdx \gets (JobIdx + 1)\%|\Delta_v|$; \label{update_index} \Comment{Update index}
    \Else
        \State Add $\tau_i$ to $ReadyQueue$ \label{ready_queue}
    \EndIf
  \EndFor
  \If{tasks deferred at time $t$}
     \State Move them to $ReadyQueue$ \label{defer_to_ready}
  \EndIf
  \If{$ReadyQueue \neq \emptyset$} \Comment{Dispatch next job}
     \State Schedule highest-priority job and append to $S$ \label{schedule_job} 
  \EndIf
  \State $t \gets t + \Delta t$ \label{alg:update_time}
  \If{$t \bmod H = 0$} \Comment{Reset delay pattern at hyperperiod}
     \State $JobIdx \gets 0$ \label{alg:resetjob}
  \EndIf
\EndWhile
\State \Return $S$ \label{return_seq}
\end{algorithmic}
\end{algorithm}
Following the methodology in Sec.~\ref{sec:prob_formulation}, we computed optimal job-delay $\Delta_i[N] \ \forall \tau_i\in \Gamma_U$.
This section details Alg.~\ref{alg:pfp_d} of the \emph{SecureRT} framework, which detects SBA attempts on the control task and switches to the \emph{PFP-d} scheduler.
Once an anomaly is detected, CPSs typically initiate a temporary mitigation phase that uses lightweight defence mechanisms to preserve the operational continuity for a certain duration~\cite{giraldo2018survey}. After this interval, a \emph{system-level reset} is typically done to restore nominal operation and enable more resource-intensive protection mechanisms~\cite{banerjee2022secure,lu2024recovery}. We refer to this parameter as $T_{RES}$, which represents the \emph{system reset time}, a predefined duration during which the runtime mitigation remains active following the detection of an SBA.
\\
$\bullet$ \textbf{\textit{Runtime Operation:} }The algorithm takes the following inputs: \textbf{(i)} task set $\mathcal{T}=(\Gamma_C \cup \Gamma_{NC}), \ \Gamma_U \subseteq \Gamma_{NC}$ with specifications, \textbf{(ii)} optimal job-level delay arrays $\Delta_i[H/T_i], \, \forall \tau_i \in \Gamma_C$, and \textbf{(iii)} detector thresholds $Th_i$, \textbf{(iv)} unit time step $\Delta t$ and system reset time $T_{RES}$. 
At t = 0, the algorithm initializes the scheduler mode to PFP. The current victim job-level delay sequence is initialized as $\Delta_v = NULL$, i.e. all elements of the delay array are initially set to $0$. The task queues \textit{ReadyQueue}, \textit{DeferQueue}, and \textit{S} are initialized as empty (line~\ref{init_arrays}). The \emph{ReadyQueue} stores jobs ready for execution, \emph{DeferQueue} temporarily holds the deferred jobs of the victim task, and $S$ dispatches the highest priority job from the \textit{ReadyQueue}. 
\par During simulation, (line~\ref{while_loop}), the \textsc{UpdateDetector()} function fetches the $\chi^2$ statistic $g_i[k]$ from all residue-based detectors associated with the control tasks (line~\ref{update_detector}). If all $\chi^2$ statistics remain below their respective thresholds, i.e. $g_k[i]\leq Th_i, \, \forall \tau_i \in \Gamma_C$, the system continues job scheduling with the PFP scheduler (as all job-level delays are zero). However, when any $g_i[k]$ exceeds its threshold ($g_i[k] > Th_i$), an attack on task $\tau_i$ is detected at the $k$-th sample. The $AtkFlag$ then returns the index of the compromised control task $\tau_i$, marking it as the victim task $\tau_v$ (see line~\ref{switch_mode_kth}).

\par After detection of an SBA on the victim control task $\tau_i$, the precomputed job-level delay sequence $\Delta_i$ is fetched from memory. 
Each task $\tau_i \in \Gamma_C$ has its own optimal job-level delay sequence $\Delta_i = [\delta_{i,1}, \delta_{i,2}, \dots, \delta_{i, H/T_i}]$, which is stored offline. The scheduler mode is switched from PFP to \emph{PFP-d} (see line~\ref{switch_mode_kth}). Upon switching, the algorithm sets the current job-level delay index to $\textit{JobIdx} = k \bmod |\Delta_i|$, corresponding to the sampling instant the attack was detected. After this, the current active job of $\tau_v$ is deferred by its corresponding job-level delay $\Delta_v[JobIdx]$ (see line~\ref{defer_victim}). The deferral operation delays the release of the active job, ensuring it follows the correct cyclic job-level delay value from the sequence $\Delta_v$.
The deferred job is placed in the \textit{DeferQueue}, which holds the pending jobs of $\tau_v$ with modified release time. Subsequent jobs of $\tau_v$ follow the remaining delays in $\Delta_v$ cyclically across hyperperiods. All other arriving tasks are directly added to the \textit{ReadyQueue}. When the system clock reaches the deferred release time, the scheduler automatically transfers the task from the \textit{DeferQueue} to the \textit{ReadyQueue} (see lines~\ref{ready_queue}–\ref{defer_to_ready}).
\par At each unit time step, the highest priority job in the \textit{ReadyQueue} is appended to $S$. (line~\ref{schedule_job}). The system clock is incremented by $\Delta t$ after each iteration, and the job-level delay index is reset at the start of every new hyperperiod. (see lines~\ref{alg:update_time} to ~\ref{alg:resetjob}). Using Alg.~\ref{alg:pfp_d}, \emph{SecureRT} updates the job release time of the victim task in real-time.
Finally, all jobs are dispatched from $S$ until the system reset time $T_{RES}$, after which a system-level reset is triggered (line~\ref{return_seq}).
\\
$\bullet$ \noindent \textbf{\textit{Correctness and Complexity:} } Since all job-level delays $\Delta_i[N_i]$ are optimally chosen from $(0,\Delta_i^{max}]$, naturally they ensure minimum overlap and schedulability.
All operations in Alg.~\ref{alg:pfp_d}, such as job deferral and ready queue updates, are performed in constant time; therefore, the complexity of the proposed \emph{PFP-d} algorithm is $O(1)$. Since all optimal job-level delay sequences, i.e. $\Delta_i[N_i], \forall \tau_i \in \Gamma_C$ are stored offline, there is negligible runtime overhead, making it suitable for online deployment in resource-constrained real-time embedded platforms.
\section{Experimental Evaluation}
\begin{figure*}[!htbp]
    \centering
    \begin{subfigure}[b]{0.23\textwidth}
        \centering
    \includegraphics[trim=0cm 4cm 0cm 0cm, width=\textwidth,clip]{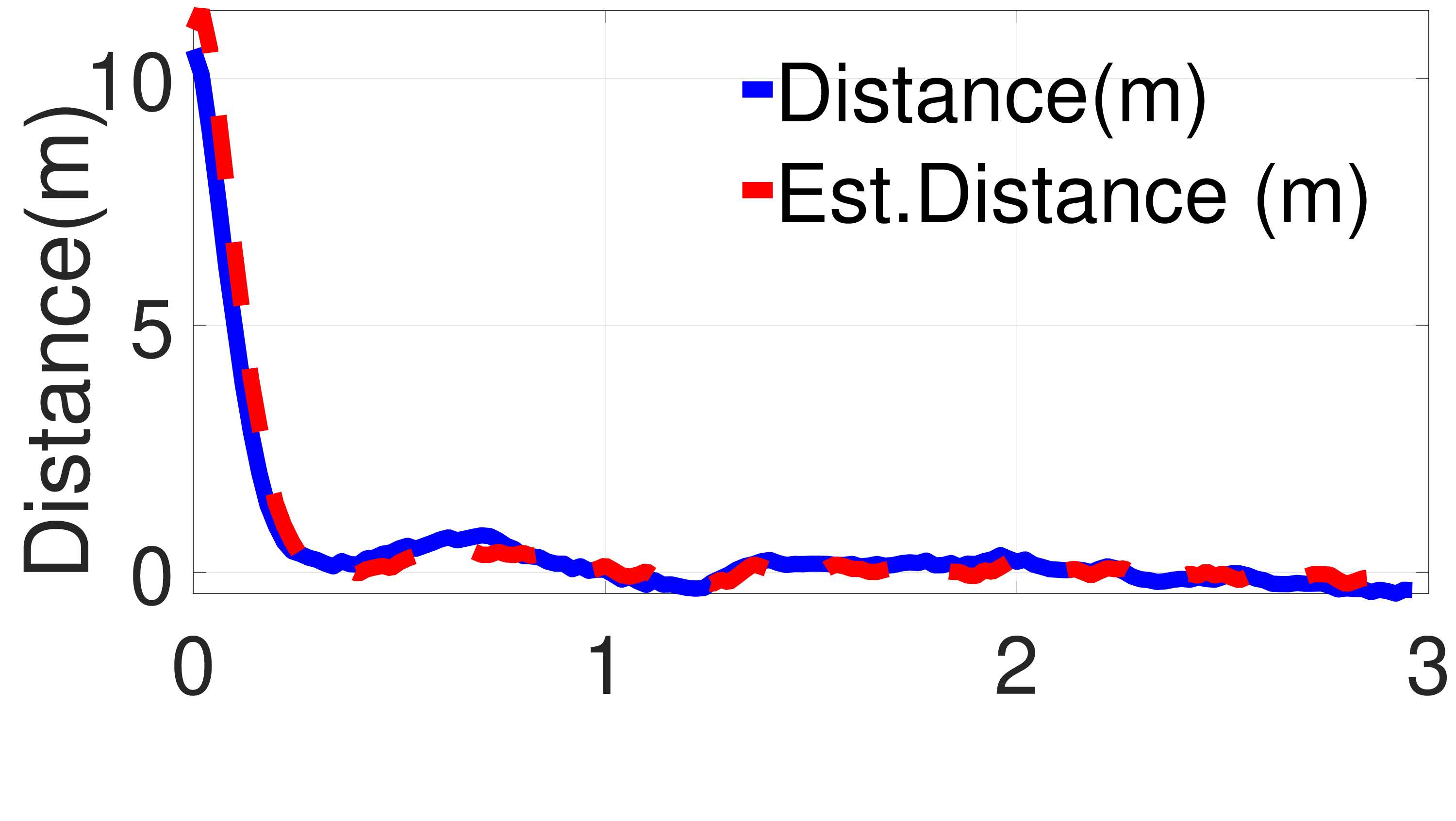}
        \caption{Case \textbf{\emph{(i)}} 
        Plant State}
        \label{fig:1_state}
    \end{subfigure}
    \hfill
    \begin{subfigure}[b]{0.23\textwidth}
        \centering
     \includegraphics[trim=0cm 4cm 0cm 0cm,width=\textwidth,clip]{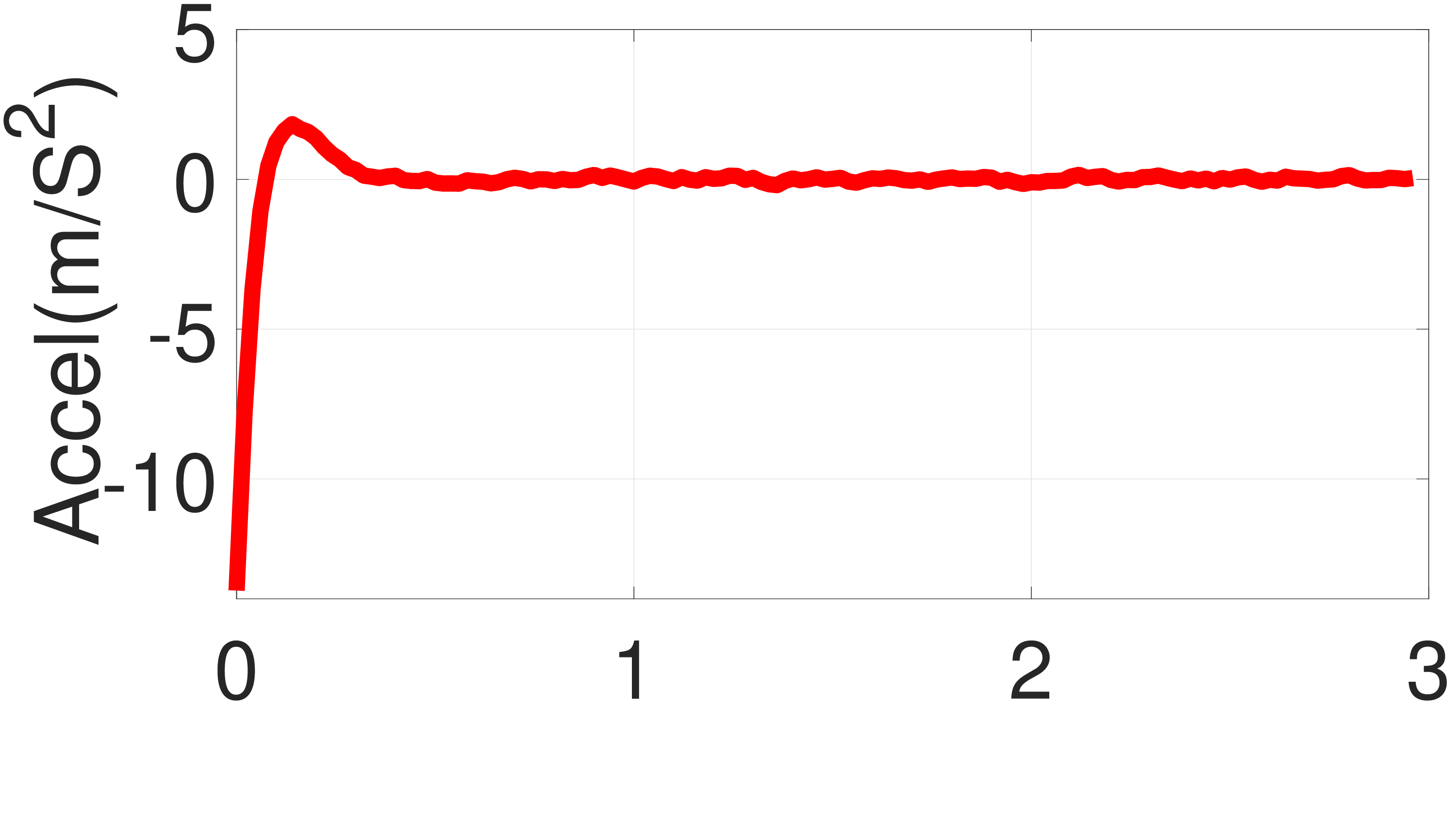}
        \caption{Case \textbf{\emph{(i)}} Control Input}
        \label{fig:1_control}
    \end{subfigure}
    \hfill
    \begin{subfigure}[b]{0.23\textwidth}
        \centering
    \includegraphics[trim=0cm 4cm 0cm 0cm,width=\textwidth,clip]{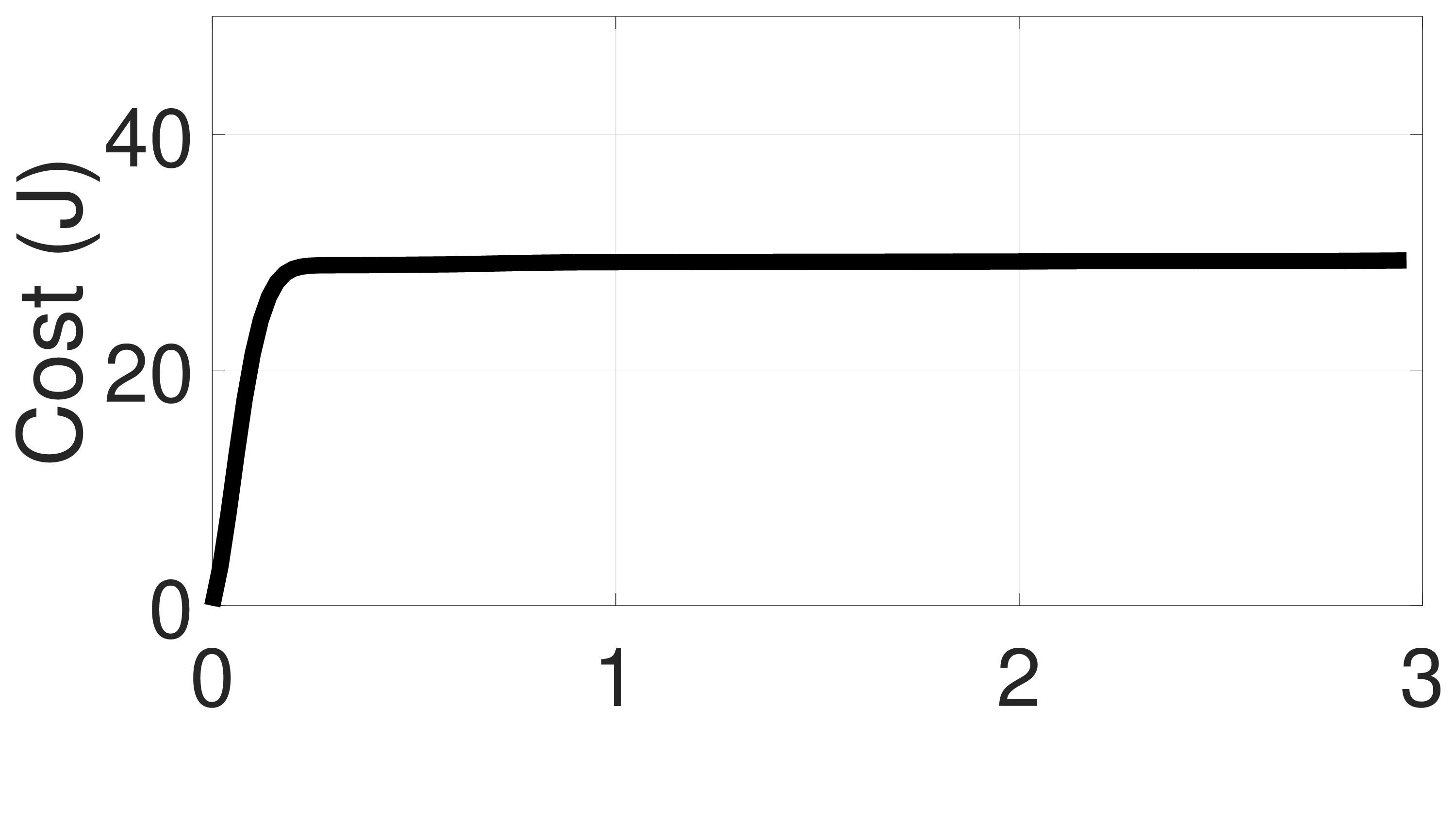}
        \caption{Case \textbf{\emph{(i)}} Control Cost}
        \label{fig:1_cost}
    \end{subfigure}
\hfill
     \begin{subfigure}[b]{0.23\textwidth}
        \centering
      \includegraphics[trim=0cm 4cm 0cm 0cm,width=\textwidth,clip]{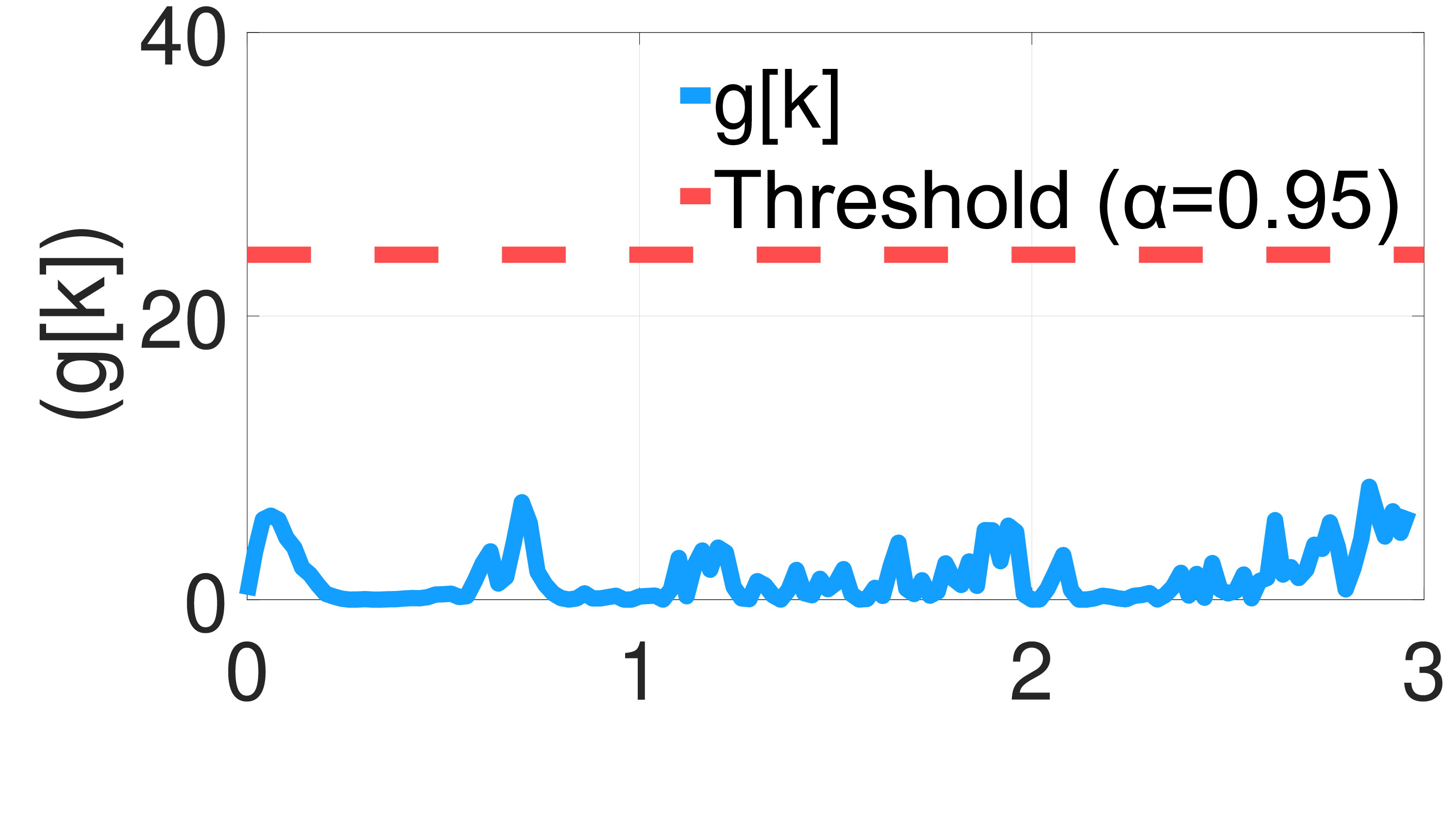}
        \caption{Case \textbf{\emph{(i)}} Residue Statistics}
        \label{fig:1_res}
    \end{subfigure}
    \begin{subfigure}[b]{0.23\textwidth}
        \centering
     \includegraphics[trim=0cm 4cm 0cm 0cm,width=\textwidth,clip]{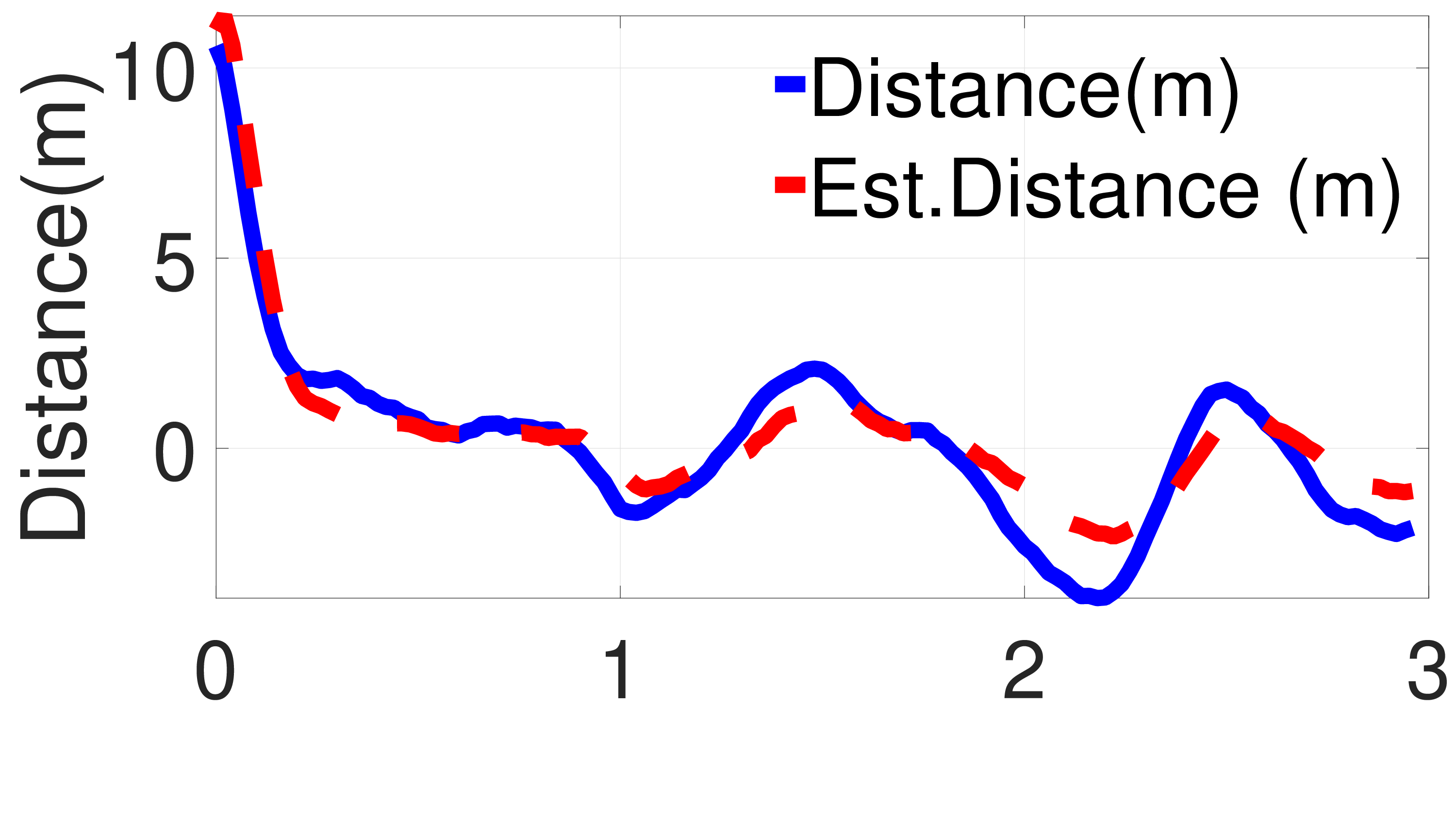}
        \caption{Case \textbf{\emph{(ii)}} Plant State}
         \label{fig:2_state}
    \end{subfigure}
    \hfill
    \begin{subfigure}[b]{0.23\textwidth}
        \centering
        \includegraphics[trim=0cm 4cm 0cm 0cm,width=\textwidth]{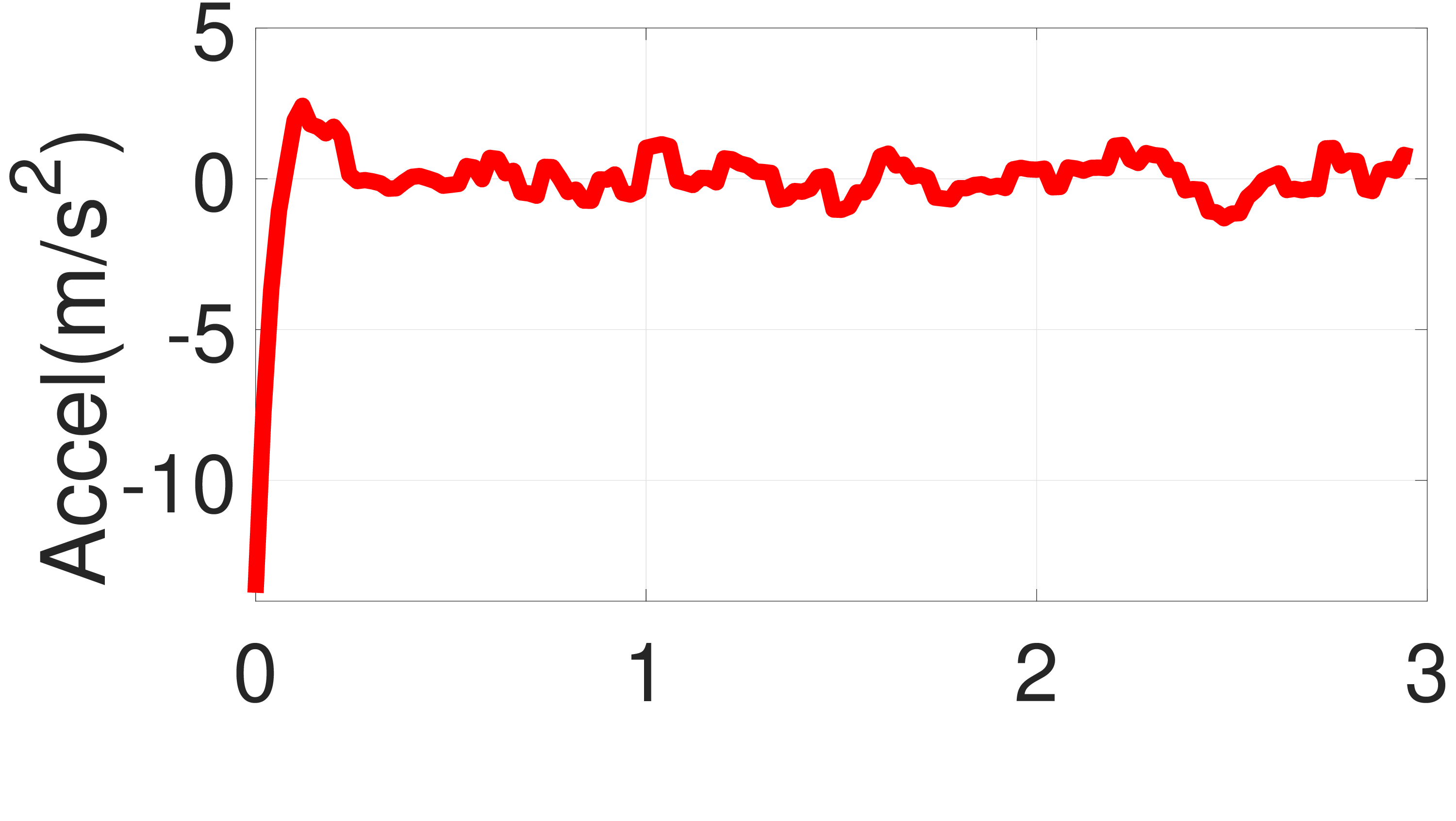}
        \caption{Case \textbf{\emph{(ii)}} Control Input}
         \label{fig:2_control}
    \end{subfigure}
    \hfill
    \begin{subfigure}[b]{0.23\textwidth}
        \centering
    \includegraphics[trim={0cm 4cm 0cm 0cm},width=\textwidth,clip]{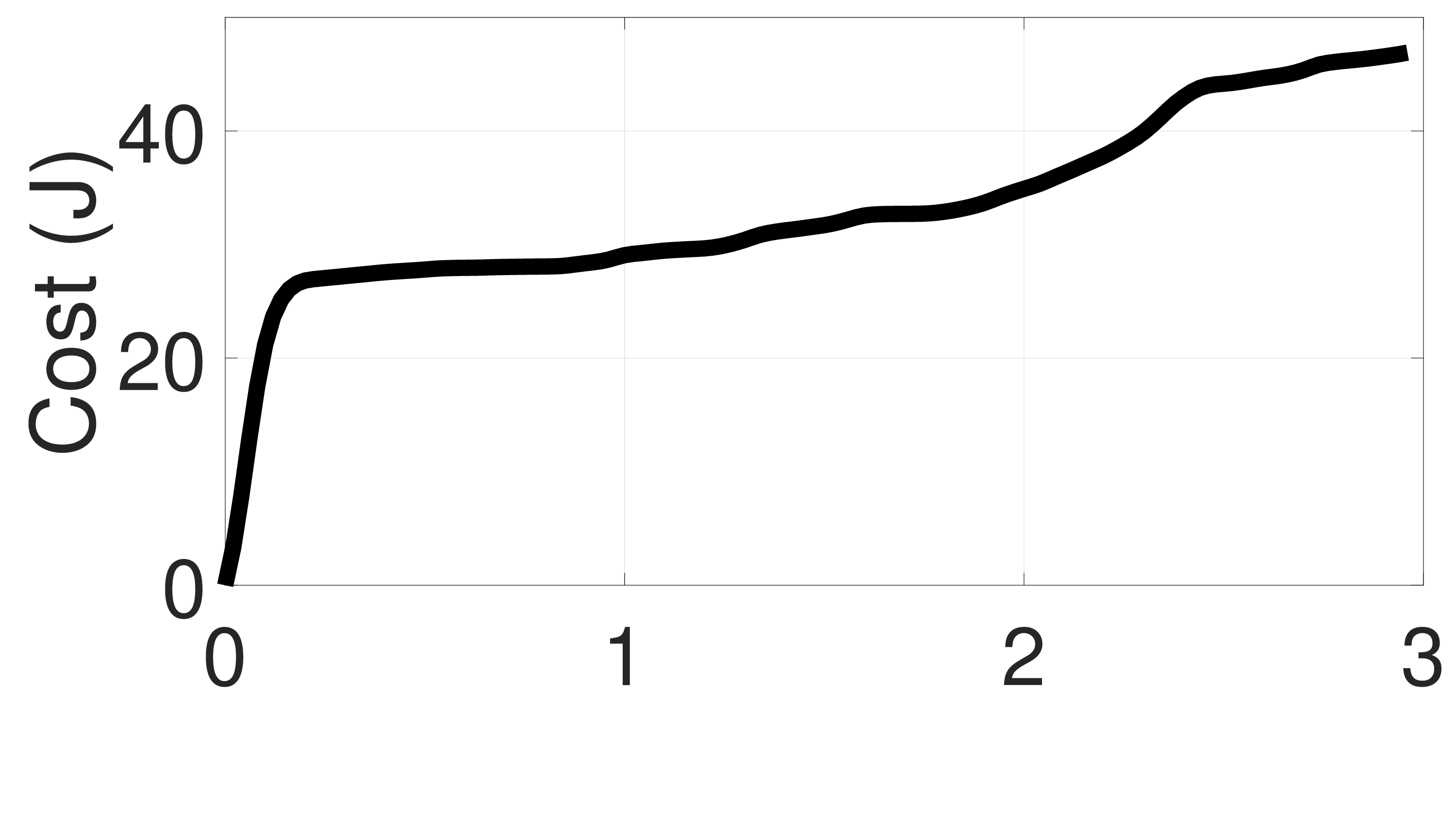}
        \caption{Case \textbf{\emph{(ii)}} Control Cost}
        \label{fig:2_cost}
    \end{subfigure}
     \hfill
    \begin{subfigure}[b]{0.23\textwidth}
        \centering
        \includegraphics[trim=0cm 4cm 0cm 0cm,width=\textwidth,clip]{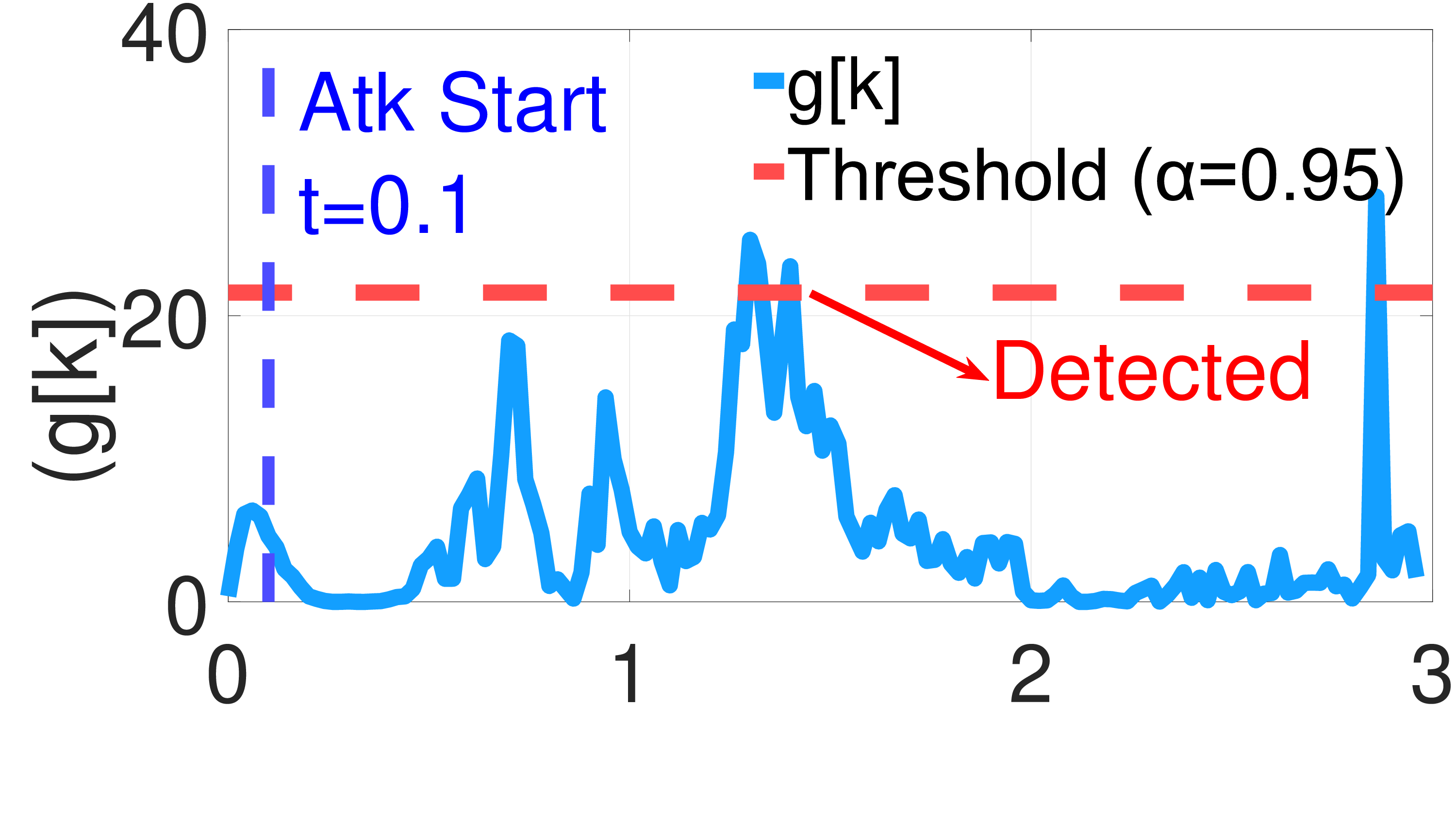}
        \caption{Case \textbf{\emph{(ii)}} Residue Statistics}
        \label{fig:2_res}
    \end{subfigure}
    \begin{subfigure}[b]{0.23\textwidth}
        \centering
        \includegraphics[trim={0cm 4cm 0cm 0cm},width=\textwidth,clip]{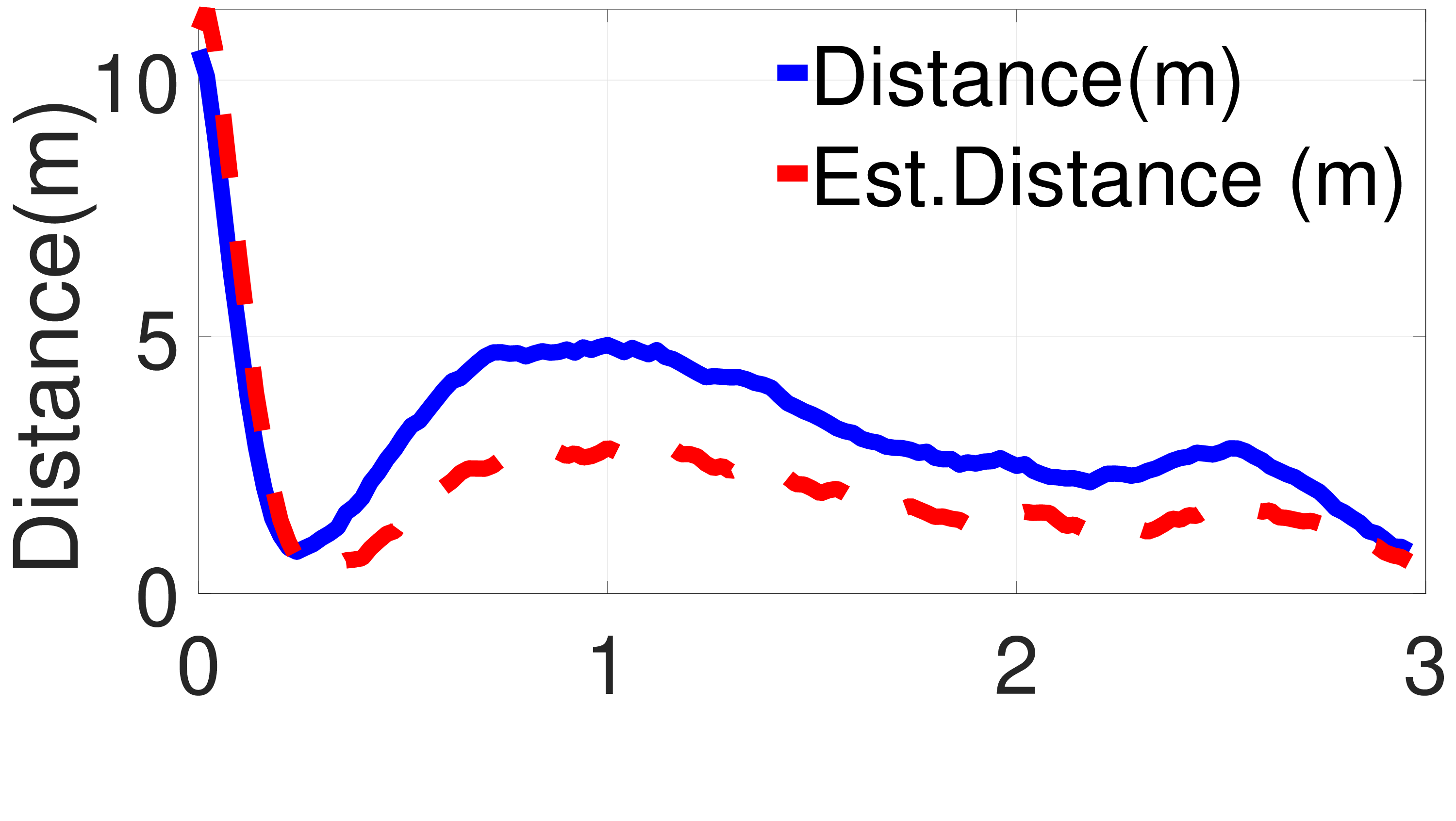}
        \caption{Case \textbf{\emph{(iii)}} Plant State}
         \label{fig:3_state}
    \end{subfigure}
    \hfill
    \begin{subfigure}[b]{0.23\textwidth}
        \centering
        \includegraphics[trim={0cm 4cm 0cm 0cm},width=\textwidth,clip]{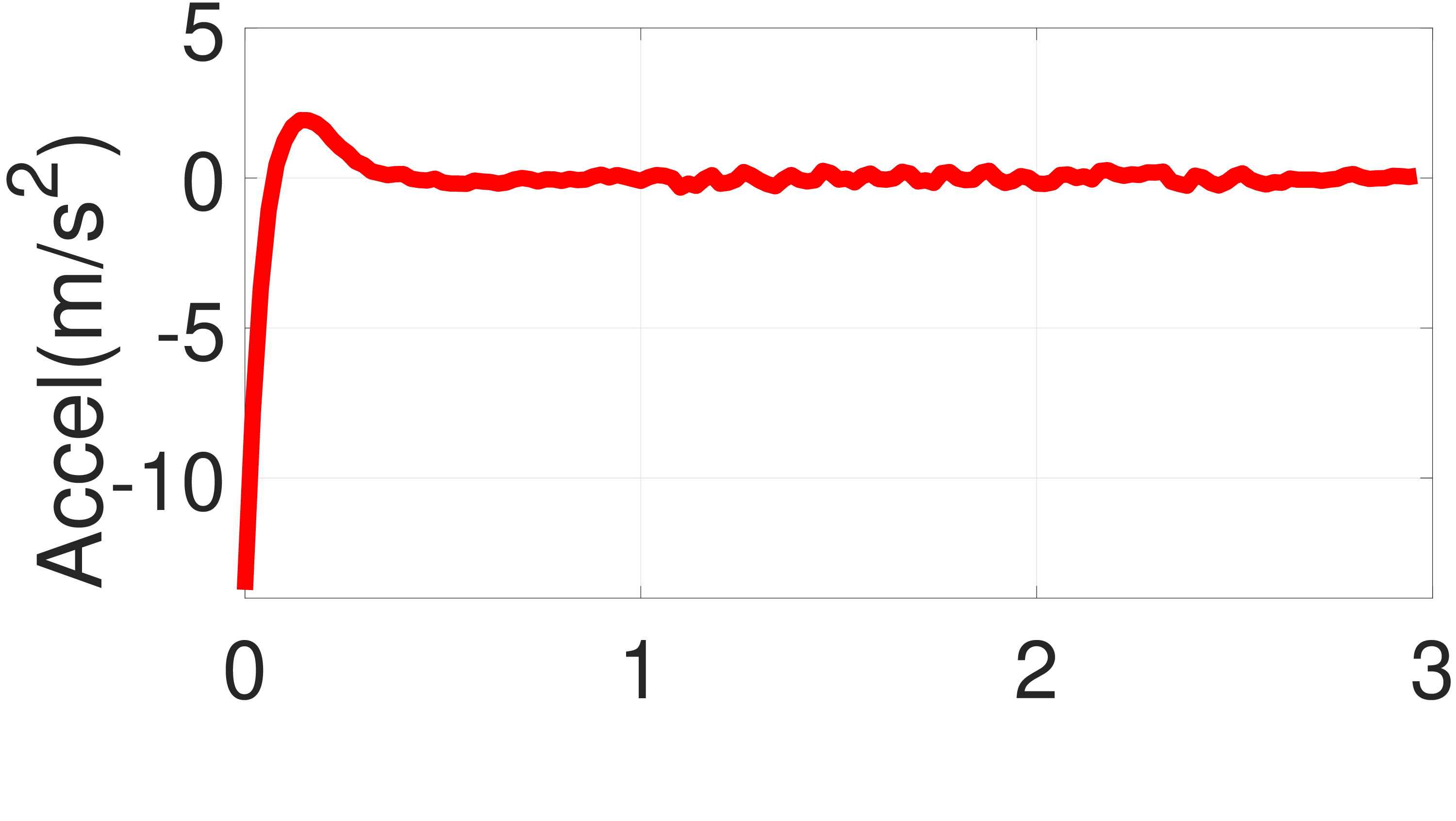}
        \caption{Case \textbf{\emph{(iii)}} Control Input}
         \label{fig:3_control}
    \end{subfigure}
    \hfill
    \begin{subfigure}[b]{0.23\textwidth}
        \centering
        \includegraphics[trim={0cm 4cm 0cm 0cm},width=\textwidth,clip]{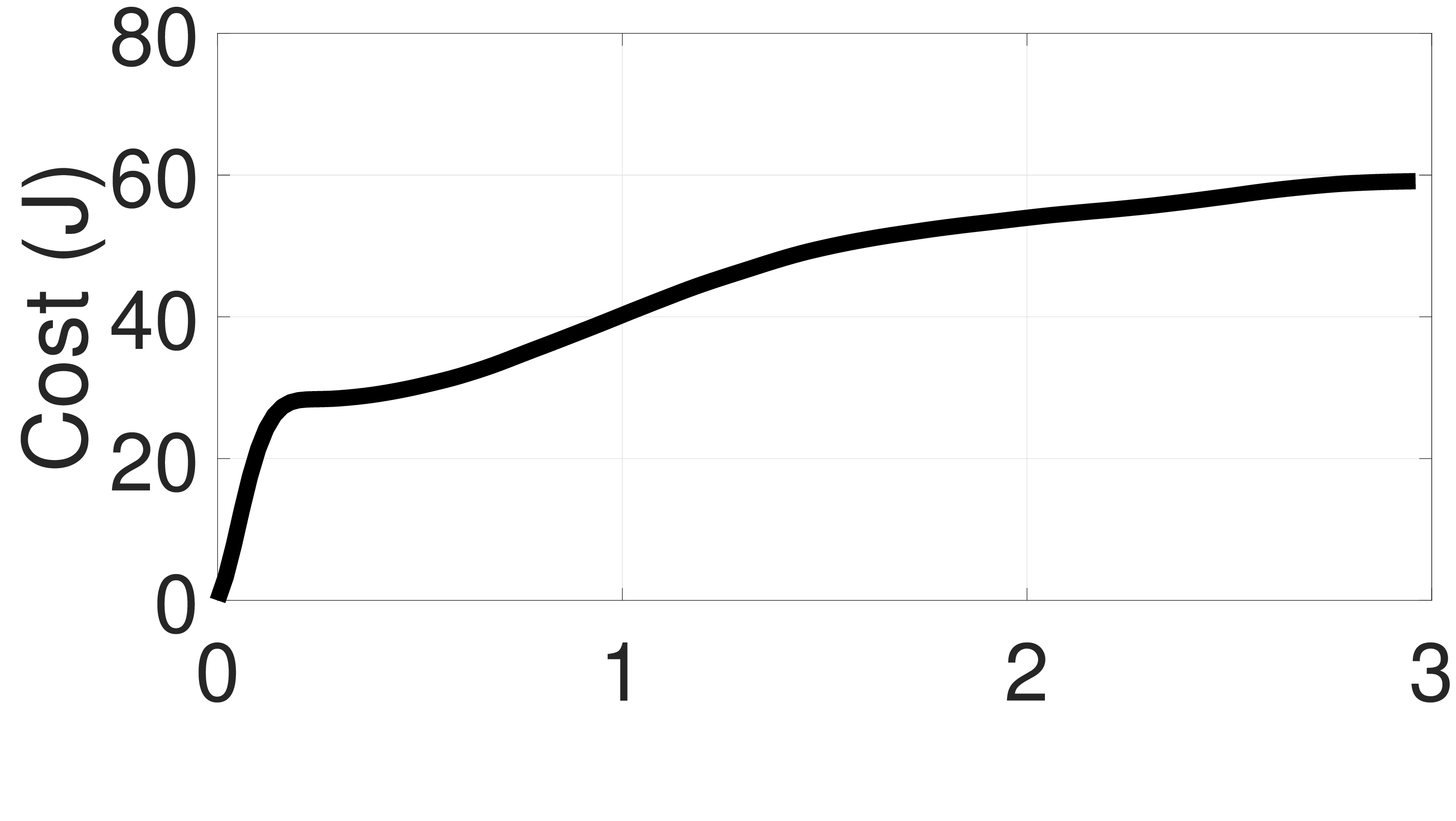}
        \caption{Case \textbf{\emph{(iii)}} Control Cost}
        \label{fig:3_cost}
    \end{subfigure}
     \hfill
    \begin{subfigure}[b]{0.23\textwidth}
        \centering
        \includegraphics[trim={0cm 4cm 0cm 0cm},width=\textwidth,clip]{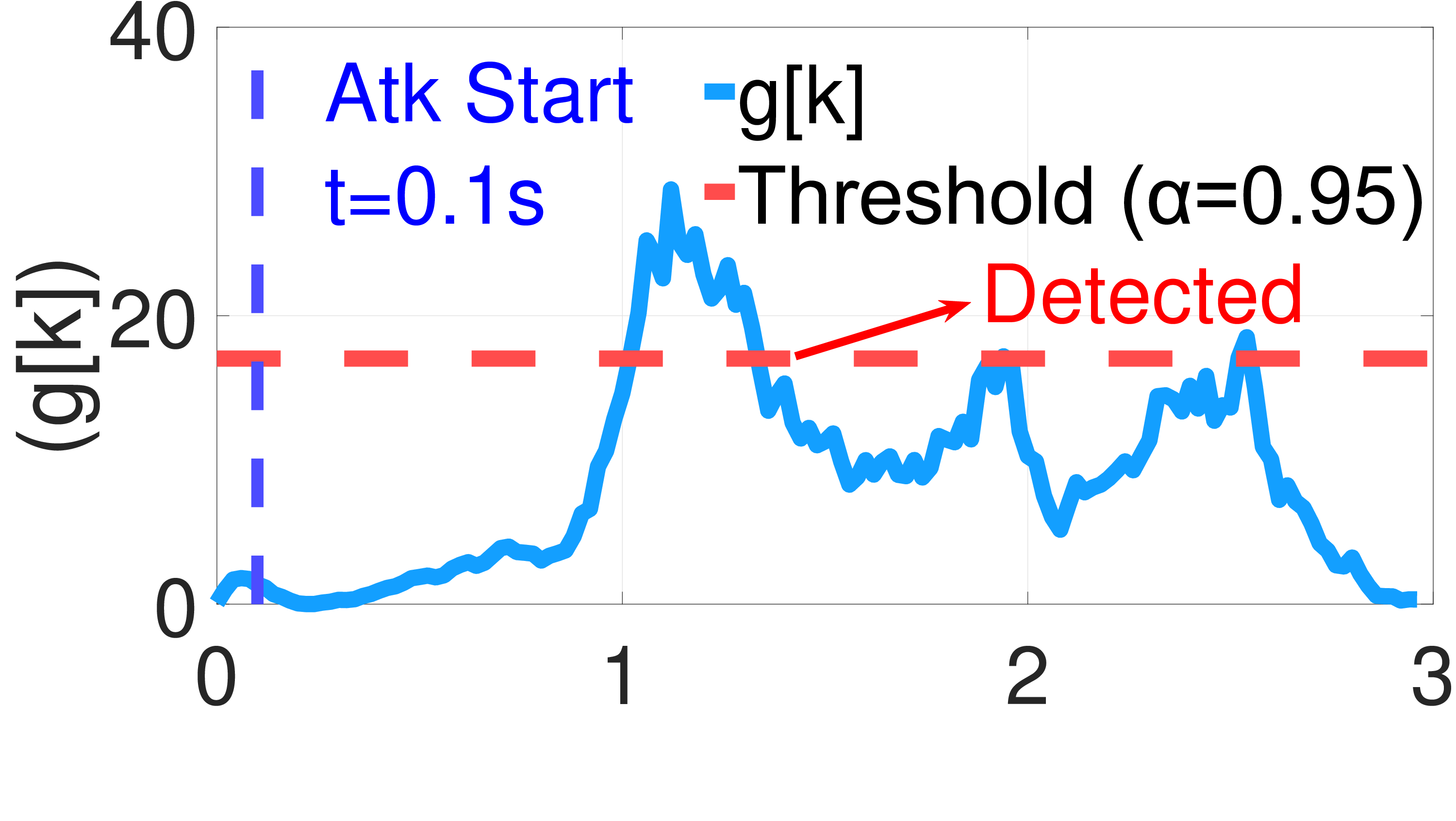}
        \caption{Case \textbf{\emph{(iii)}} Residue Statistics}
        \label{fig:3_res}
    \end{subfigure}
    \begin{subfigure}[b]{0.23\textwidth}
        \centering
        \includegraphics[trim={0cm 2.2cm 0cm 0cm},width=\textwidth,clip]{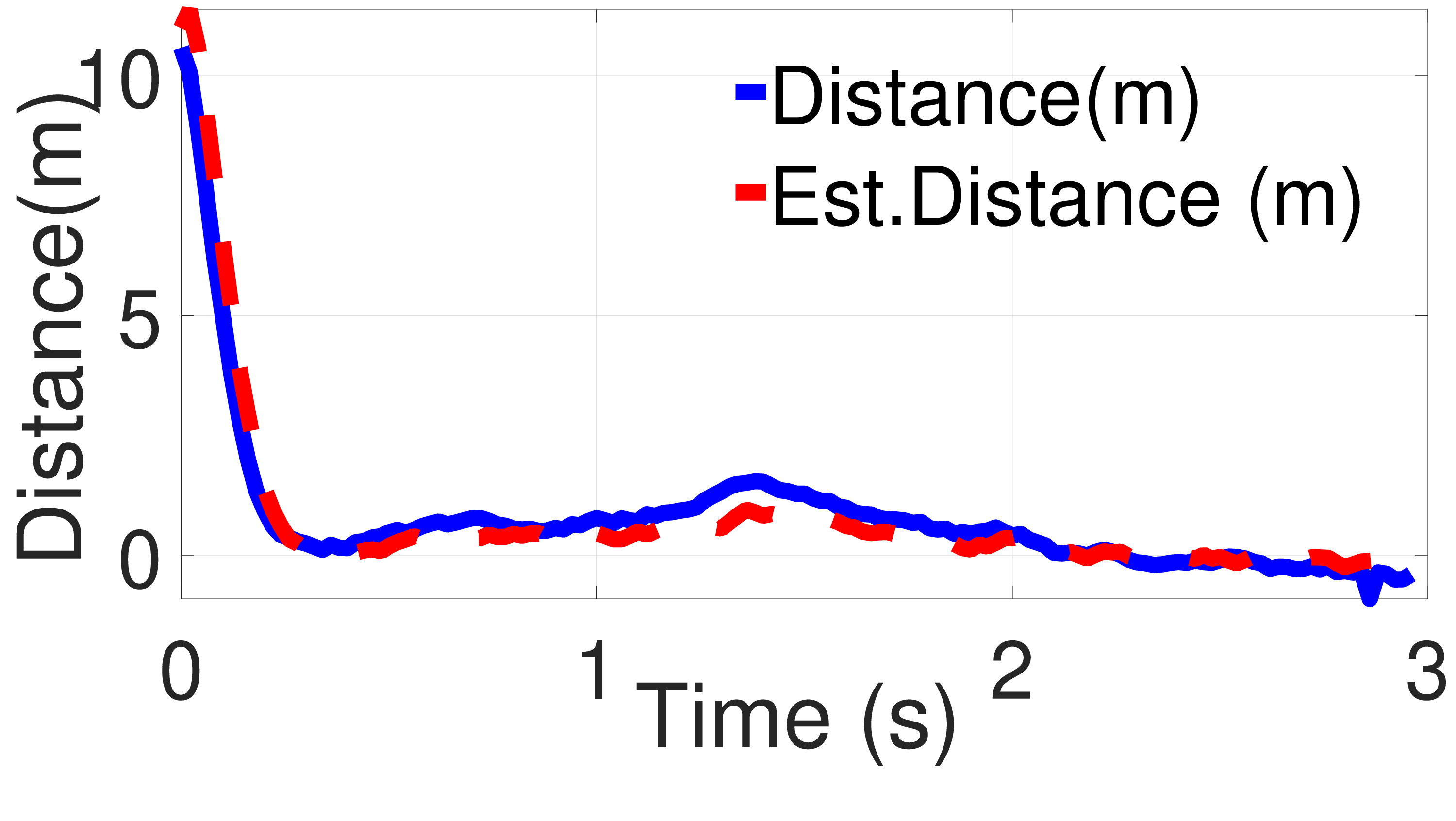}
        \caption{Case \textbf{\emph{(iv)}} Plant State}
        \label{fig:4_state}
    \end{subfigure}
    \hfill
    \begin{subfigure}[b]{0.23\textwidth}
        \centering
        \includegraphics[trim={0cm 1.9cm 0cm 0cm},width=\textwidth,clip]{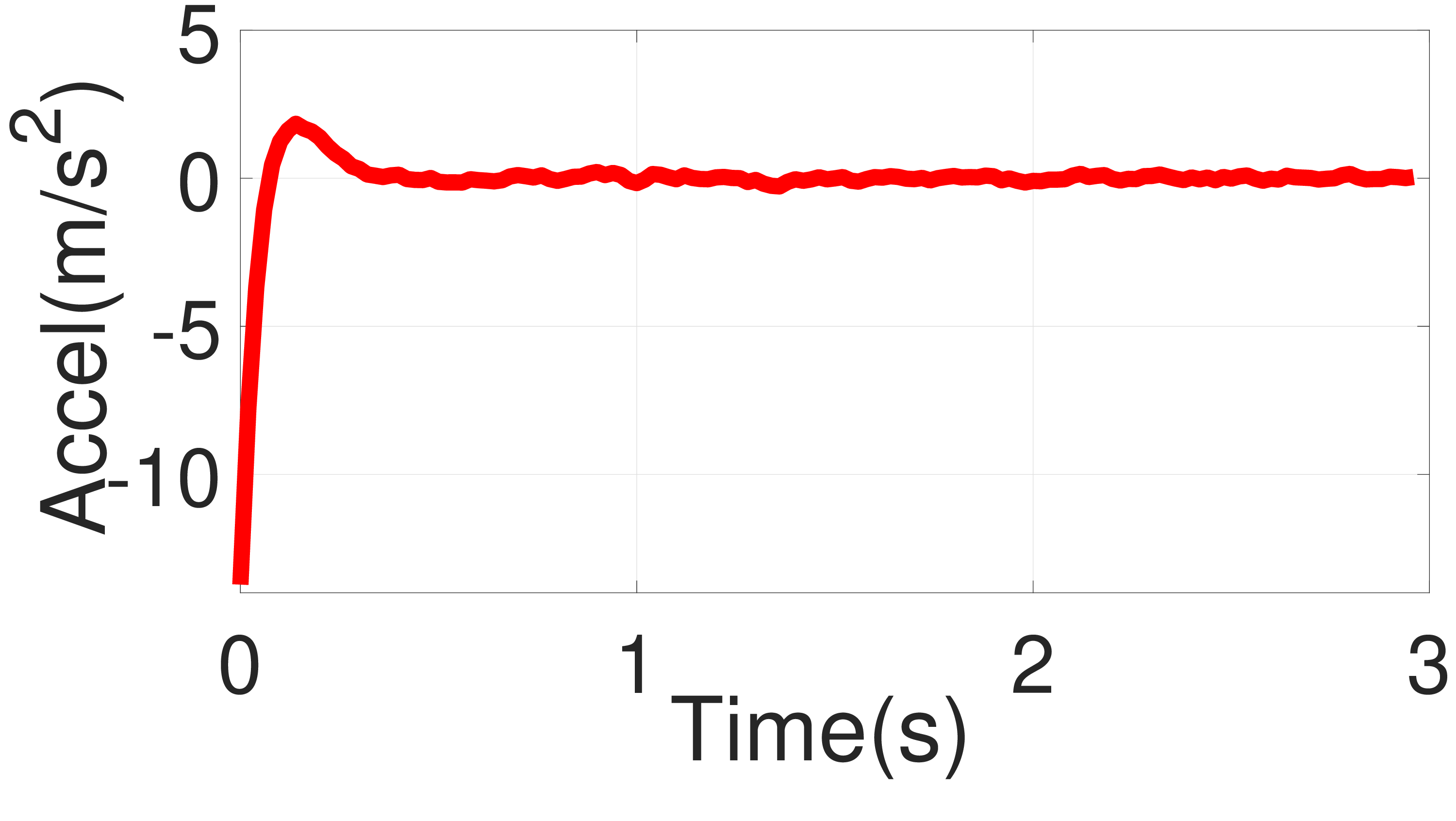}
        \caption{Case \textbf{\emph{(iv)}} Control Input}
        \label{fig:4_control}
    \end{subfigure}
    \hfill
    \begin{subfigure}[b]{0.23\textwidth}
        \centering
        \includegraphics[trim={0cm 1.8cm 0cm 0cm},width=\textwidth,clip]{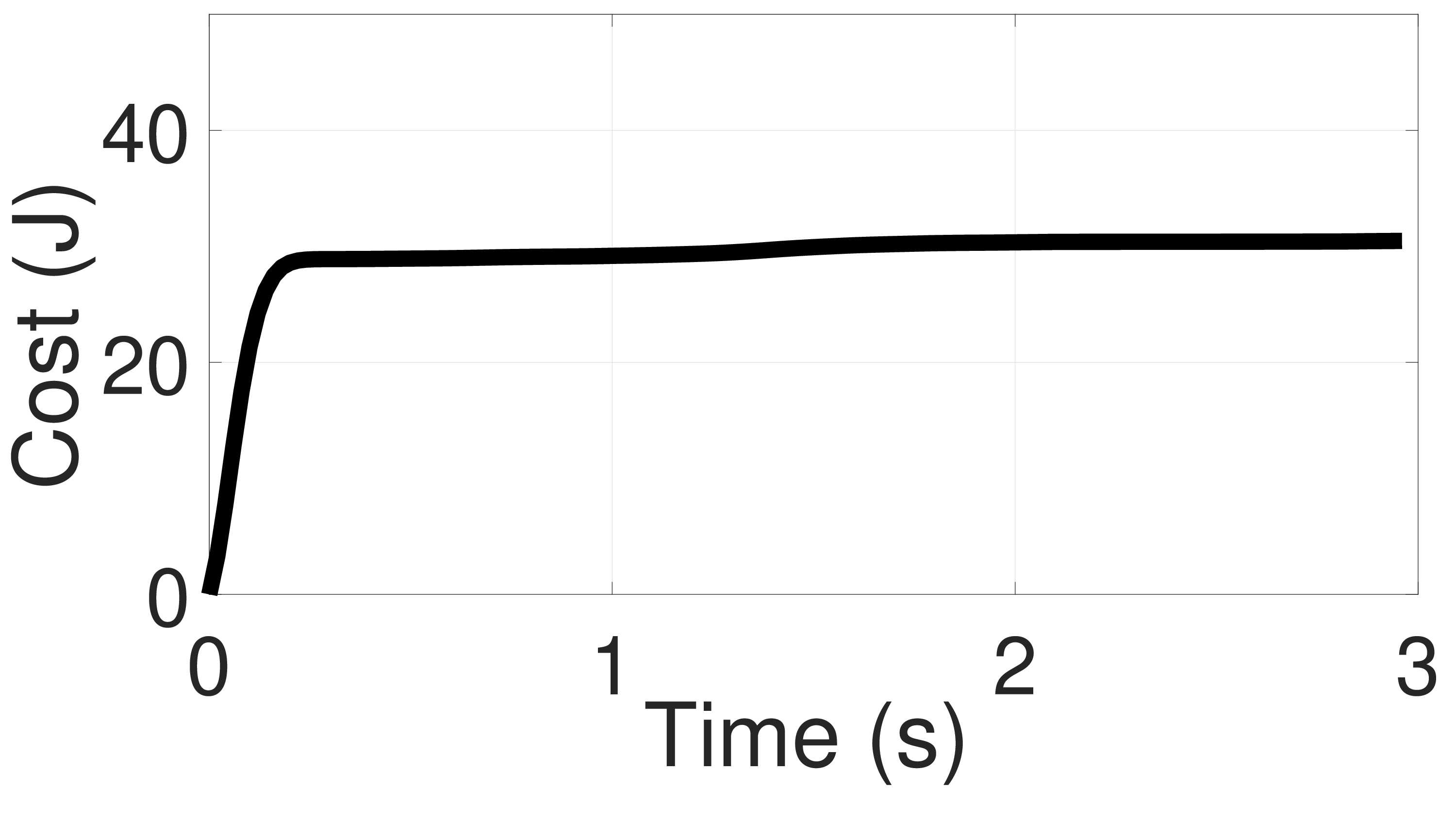}
        \caption{Case \textbf{\emph{(iv)}} Control Cost}
        \label{fig:4_cost}
    \end{subfigure}
    \hfill
    \begin{subfigure}[b]{0.23\textwidth}
        \centering
        \includegraphics[trim={0cm 1.8cm 0cm 0cm},width=\textwidth,clip]{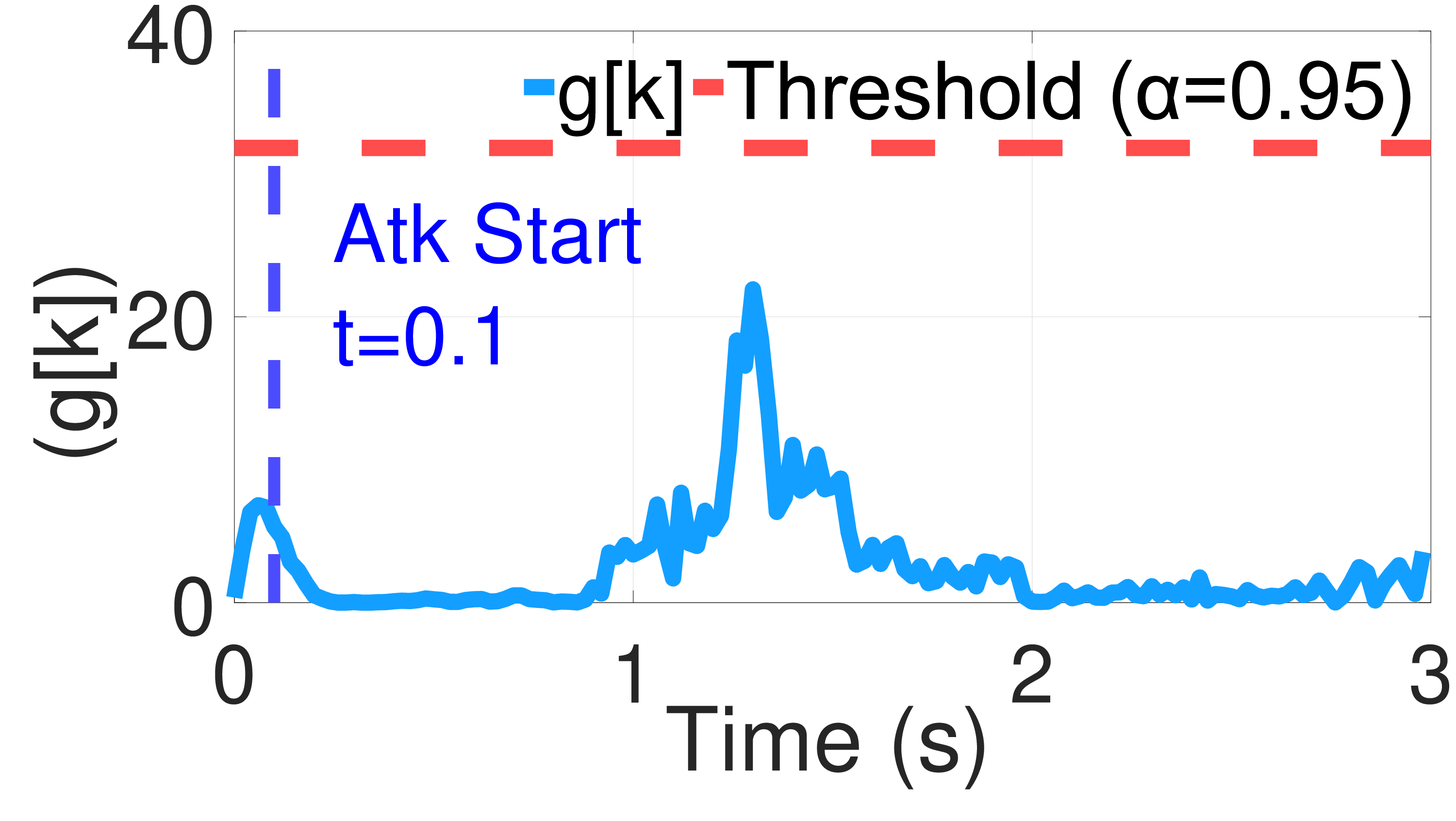}
        \caption{Case \textbf{\emph{(iv)}} Residue Statistics}
        \label{fig:4_res}
    \end{subfigure}
    \vspace{-1mm}
    \caption{Comparison of four runtime cases: (i) baseline PFP without attack, (ii) PFP under posterior SBA, (iii) random job-level delays under posterior SBA, and (iv) \emph{SecureRT} under posterior SBA, showing plant state, control input, cost, and $\chi^2$ residue.}
    \label{fig:control_evaluation}
    \vspace{-7mm}
\end{figure*}
%
\par $\bullet$ \textbf{\textit{Real-time Experimental Setup: }}We demonstrate the applicability of {\em SecureRT} by using a realistic case study consisting of an automotive control task set similar to the works in~\cite{koley2023cad,nasri2019pitfalls,sain2025maars}. We implemented three distinct automotive controllers for our evaluation, namely \emph{\textbf{(i)}} a cruise controller (CC, ensures desired vehicle speed by actuating proper acceleration), \emph{\textbf{(ii)}} an electronic stability program (ESP, manipulates yaw angles by controlling steering angle), and 
\emph{\textbf{(iii)}} a trajectory tracking controller (TTC, regulates deviation from a desired trajectory with acceleration inputs). The design parameters 
of control tasks $\{ \tau_1, \tau_2,\tau_3\}$ corresponding to these controllers, along with other non-control tasks $\{ \tau_4, \tau_5,\tau_6\}$, are tabulated in Tab.~\ref {tab:casestudy}. 
\begin{table}[ht]
\centering
\renewcommand{\arraystretch}{1.2}
\begin{tabular}{|c||c|c|c|c|c|c|c|}
\hline
Task & Task & $T_i$& $C_i$ & $D_i$ & $\Omega_i$  & $\Delta_i^{peak}$ & $\Delta_i^{\text{max}}$ \\
 & Type & [ms] & [ms] & [ms] & [ms]  & [ms]& [ms] \\
\hline
\hline
$\tau_1$ & CC & 10  & 2 & 10  & 8  & 8 & 3 \\
$\tau_2$ & ESP & 40  & 3 & 40  & 7  & 35  & 12 \\
$\tau_3$ & TTC  & 20  & 2 & 20  & 5 & 13  &  8 \\
$\tau_4$ & --   & 100 & 5 & 100 & -- & -- & -- \\
$\tau_5$ &  --  & 100 & 4 & 100 & -- & -- & -- \\
$\tau_6$ &  --  & 40  & 2 & 40  & -- & -- & -- \\
\hline
\end{tabular}
\caption{Automotive Case-Study Task Set}
\vspace{-4mm}
\label{tab:casestudy}
\end{table}
Typically, non-control tasks perform non-critical operations, such as data logging, system supervision, and diagnostic monitoring. These are considered untrusted tasks $\Gamma_U= \{\tau_4, \tau_5,\tau_6\}$, whereas the set of control tasks $\Gamma_C= \{\tau_1,\tau_2,\tau_3\}$ are considered victim tasks. 
%
\par The taskset is scheduled
on a quad-core ARM Cortex-A72 processor in a {\em Raspberry Pi 4 Model B} with 4 GB of RAM, running Raspberry Pi OS (64-bit, kernel v6.12) based on Debian~Bookworm, a linux variant.
Note that the Linux kernel does not inherently provide hard real-time guarantees, even when using real-time thread scheduling policies such as \texttt{SCHED\_FIFO}, \texttt{SCHED\_RR}, or \texttt{SCHED\_DEADLINE} with the POSIX threads (pthreads) API. Therefore, to achieve real-time uniprocessor capabilities, such as an automotive ECU, we make the following changes in the default Raspberry Pi Setup. {\bf(i)} We patched the kernel with \texttt{PREEMPT\_RT} (version~6.12-rt), from Real-Time Linux (RTL) Collaborative Project~\cite{rtl_wiki}.
{\bf(ii)} During the kernel compilation, we disabled the \texttt{CONFIG\_SMP} flag to limit execution on a single core. {\bf(iii)} To utilize 100\% of the CPU time, the variable \texttt{sched\_rt\_runtime\_us} was set to -1. {\bf(iv)} We disabled the CPU frequency scaling governor and set the processor to a fixed clock speed of 1\,GHz using the linux \texttt{cpupower} utility. 
\par \noindent $\bullet$\textit{\textbf{ Design-time Analysis: }} First, we compute the peak job-level delay $\Delta_i^{peak}$ for the control tasks in Tab.~\ref{tab:casestudy} by using Eq.~\ref{eq:peakdelay} (see Sec.~\ref{sec:peak-delay}). These are tabulated in column 7 of Tab.~\ref{tab:casestudy}. 
To determine $\Delta_i^{max}$, we first set the control cost threshold $J_i^{Th}$ at $5\%$ above the nominal cost without any job-level delay. We then model and simulate a delay-aware closed-loop variant of all three control tasks in $\Gamma_C$ in MATLAB R2024b, with incremental job-level delays $\delta \in [0, \Delta_i^{peak}]$. For each delay value, the corresponding control cost $J_i(R_i(\delta))$ was computed using Eq.~\ref{eq:lqr_cost}. Finally, Eq.~\ref{eq:maxdelay} was used to determine the maximum admissible delay $\Delta_i^{max}$. The maximum admissible job-level delays are tabulated in column 8 of Tab.~\ref{tab:casestudy}. 
Given the task set $\mathcal{T}=\Gamma_C \cup \Gamma_U$, their parameters, along with the maximum admissible delays $\Delta_i^{max}$ and AEWs ($\Omega_i$), we compute the optimal job-level delays $\Delta_i$ by solving the MILP formulated in Eq.~\ref{eq:objective_func_lin} using \emph{Gurobi} optimization solver. 
\par For our experiment, we choose TTC ($\tau_3$) as the victim control task. The MILP features $280$ auxiliary, $180$ binary variables, and $1080$ constraints (excluding lower bounds in Constraint~\ref{constr:C1}) to derive the optimal job-level delays for $\tau_3$. The optimization was performed on a 16-core AMD Ryzen 9 8945HS laptop with 16 GB RAM, and the optimal value of $\Delta_3$ was obtained within 1.34 seconds. For $\tau_3$, we obtained the job-level delay sequence $\Delta_3 = \{8,0,5,0,5,8,5,0,5,0\}$. Each delay from this sequence was applied to delay the job inter-arrival times of $\tau_3$ within one hyperperiod (200ms). This reduces the \emph{overlap duration} from 45ms to 18ms, promising a 60\% reduction in time available for FDI on $\tau_3$.
\par \noindent $\bullet$\textit{\textbf{ Runtime Deployment and Analysis: }}
To implement the runtime \emph{PFP-d} scheduler presented in Alg.~\ref{alg:pfp_d}, we use \texttt{SCHED\_FIFO} policy, which maintains a dedicated real-time run queue (\texttt{rt\_rq}) per CPU, always dispatching the highest-priority runnable task first. 
We create POSIX threads for tasks from Tab.~\ref{tab:casestudy} and assign priorities in order. To implement job-level delay sequence $\Delta_3$, we use the timer \texttt{clock\_nanosleep()} to defer only the $k$-th victim job ($\tau_3$) by $\delta_k \in \Delta_3$. The delayed jobs are enqueued to \texttt{rt\_rq} only after their deferral interval expires.
\par We perform our experiment with four distinct cases: \emph{\textbf{(i)}} Baseline scheduling with PFP scheduler without any SBA,\emph{\textbf{ (ii)}} PFP scheduling in the presence of a posterior SBA,
\emph{\textbf{(iii)}} scheduling with random job-level delays following state-of-the-art approach in~\cite{chen2021indistinguishability} under a posterior SBA, and \emph{\textbf{(iv)}} scheduling with the proposed \emph{PFP-d} scheduler using the \emph{SecureRT} framework under a posterior SBA. The simulation duration is $3$s for all four cases. In cases \emph{\textbf{(ii)}}, \emph{\textbf{(iii)}}, and \emph{\textbf{(iv)}}, the attacker task $\tau_4 \in \Gamma_U$ executes a malicious code to launch a posterior SBA from $t=0.1s$, which performs FDI targeting $\tau_3$'s data buffer to inject malicious control input and hamper control performance. The detector’s threshold is set to maintain a false alarm rate $\leq 5\%$.
In Fig.~\ref{fig:control_evaluation}, we have plotted (in $y$-axes) the plant's actual and estimated states (blue solid lines and red dashed lines in Col.~1, respectively), control input (in Col.~2), control cost (in Col.~3) and the $\chi^2$ statistics of the residue (in Col.~4). All plots are plotted against time in seconds ($x$-axis).
\par In case \textit{\textbf{(i)}} (No Attack, No Delay), the estimated system trajectory correctly tracks the plant state (Fig.~\ref{fig:1_state}) and the computed control input (Fig.~\ref{fig:1_control}) stabilizes the system to a steady state. The cost (Fig.~\ref{fig:1_cost}) converges to its nominal value, and the $\chi^2$ statistic $g[k]< Th$ (Fig.~\ref{fig:1_res}).
In case \textit{\textbf{(ii)}} (Posterior Attack, No Delay), the estimated plant trajectory deviates from the actual state (Fig.~\ref{fig:2_state}), due to continuous successful FDI attacks on the control input (Fig.~\ref{fig:2_control}) by $\tau_4$. This causes 
a steady increase in the control cost 
(Fig.~\ref{fig:2_cost}), and the $\chi^2$ statistic $g[k] > Th$ (Fig.~\ref{fig:2_res}).
%
In case \textit{\textbf{(iii)}} (Posterior Attack, Random Delays), we apply random job-level delays sampled from a range $[0,\Delta_v^{peak}]$ following a Laplace distribution as considered by the authors~\cite{chen2021indistinguishability} to job release times of $\tau_3$. As a result, the estimated plant trajectory deviates from the actual state (Fig.~\ref{fig:3_state}) as $\tau_4$ is able to successfully perform FDI attacks at certain instances of the control task $\tau_3$ (Fig.~\ref{fig:3_control}). This causes the $\chi^2$ statistic $g[k]$ to go beyond the threshold $Th$ (Fig.~\ref{fig:3_res}). Moreover, the control cost increases steadily (Fig.~\ref{fig:3_cost}) because the delays are not performance-aware, and FDIs are successful in certain instances.
In case \textit{\textbf{(iv)}} (Posterior Attack, with SecureRT), 
the optimal delay sequence $\Delta_3$, obtained from the design-time analysis is applied to the release times of $\tau_3$ instances. The estimated system trajectory closely tracks the plant state (Fig.~\ref{fig:4_state}). Despite the ongoing SBA on control input by $\tau_4$, the control input (Fig.~\ref{fig:4_control}) stabilizes the system to the steady state, and the overall control cost converges close to its nominal value (Fig.~\ref{fig:4_cost}). 
\par $\bullet$ \textbf{\textit{General Evaluation of Schedulability: }}
\label{sec:exp_geneval}
In our earlier experiment, we demonstrated the effectiveness of \emph{SecureRT} on an automotive task set in Tab.~\ref{tab:casestudy}, comprising high-priority control tasks that implement realistic automotive controllers. While this experiment validated its applicability for a specific real-world case, we extend our evaluation to assess its general applicability to synthetic task sets spanning multiple utilization ranges and include victim tasks from different priority groups. To perform this, we generate synthetic task sets uniformly across ten utilization ranges defined as $[0.02 + 0.1i, 0.18 + 0.1i], i = \{0, 1,...,9\}$. For each range, we created $100$ independent task sets using \emph{RandFixedSum}~\cite{emberson2010techniques}. Each task set consists of $n$ periodic tasks, where $n$ is given as a user input. The period $T_i$ of each task $\tau_i$ is randomly selected from the set $\{5, 10, 20, 50, 100, 200, 1000\}$, and its corresponding WCET $C_i$ is an integer uniformly selected from $[1, 50]$ such that $C_i < T_i$. The $n$ tasks from each set are assigned priorities using the rate-monotonic policy, so that the task with a smaller period has higher priority. We divided every task set into three groups based on their priority levels: \emph{high-priority} (HP) tasks with task priorities between $\{1,...\frac{n}{3}\}$, \emph{medium-priority} (MP) from $\{\frac{n}{3}+1,...,\frac{2n}{3}\}$, and \emph{low-priority} (LP) from $\{\frac{2n}{3}+1,...,n\}$. 
\par For each task set, we select one victim task $\tau_v$ randomly from each priority group. We incrementally vary $\delta \in [0, T_v - C_v]$ and compute WCRT of victim using 
\begin{figure}[!htbp]
    \centering
    \vspace{-2mm}
    \begin{subfigure}[t]{0.48\columnwidth}
        \centering
        \includegraphics[width=\linewidth,clip]{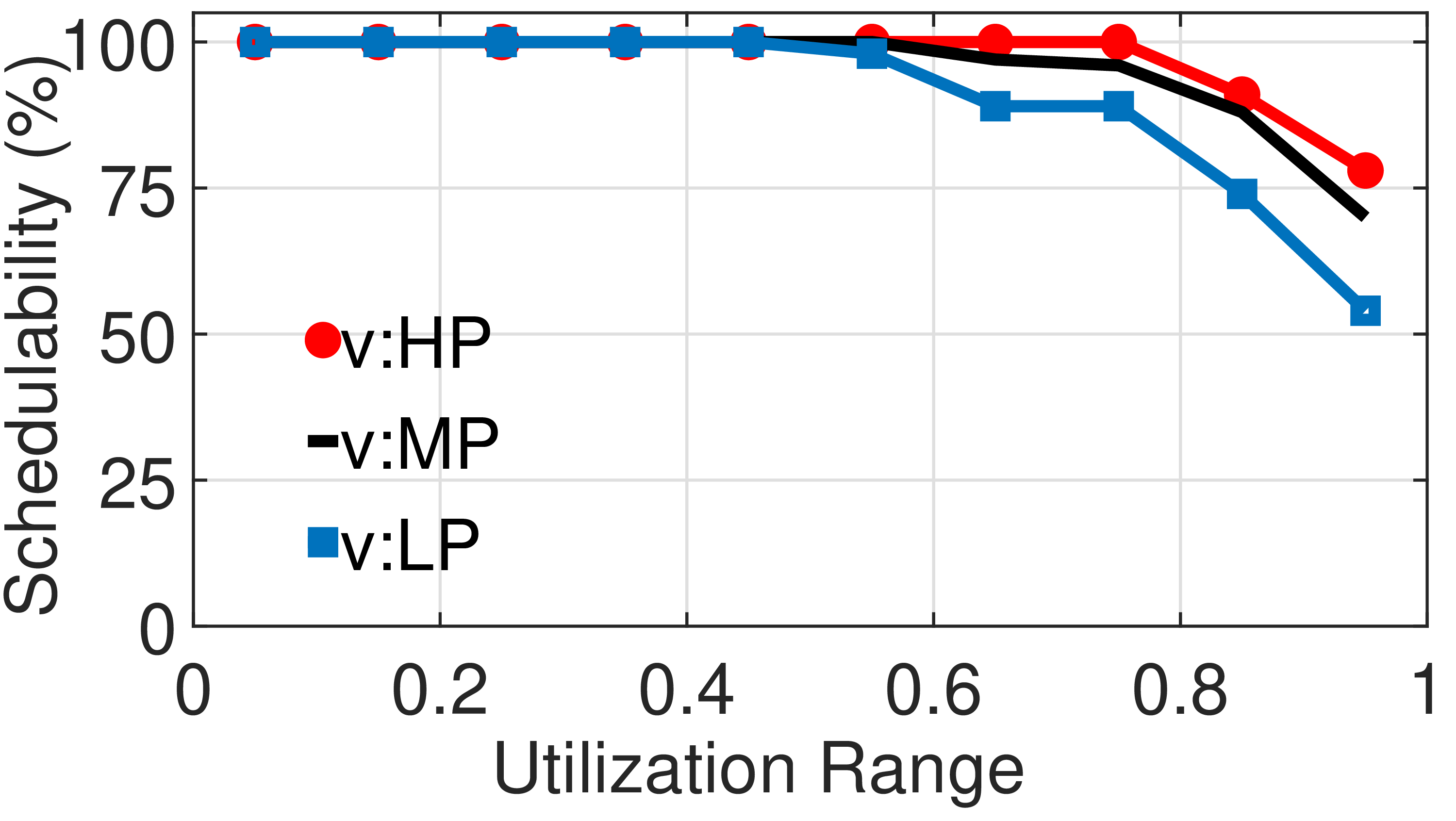}
    \end{subfigure}
    \hfill
    \begin{subfigure}[t]{0.48\columnwidth}
        \centering
        \includegraphics[width=\linewidth,clip]{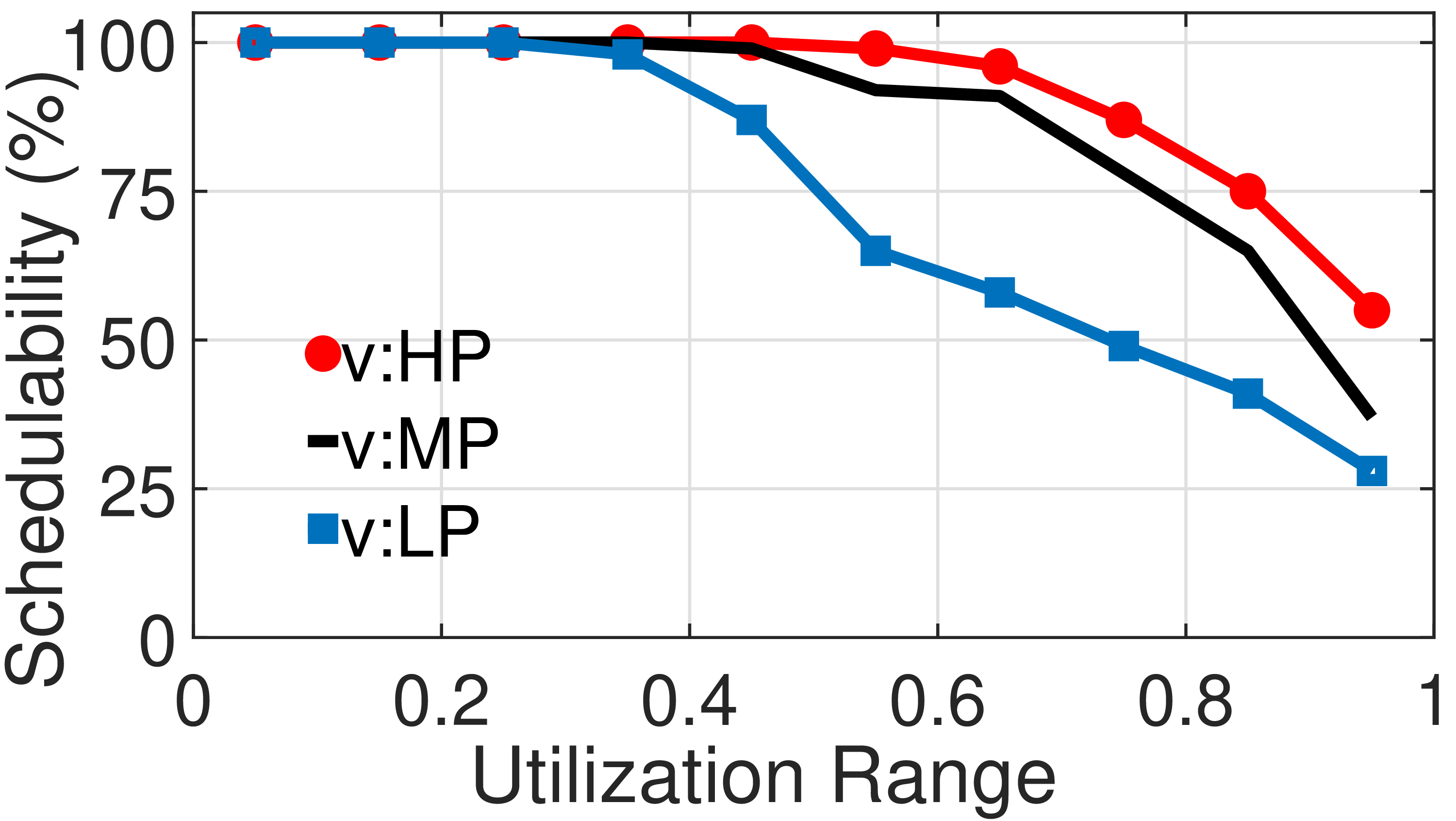}
    \end{subfigure}
    \vspace{-2mm}
    \caption{Schedulability with (a) 5 Tasks, (b) 10 tasks}
    \label{fig:schedulability_5_10}
\end{figure}
\vspace{-2mm}
\begin{wrapfigure}[8]{l}{0.48\columnwidth}
    \centering
    \vspace{-3mm}
    \includegraphics[width=\linewidth,clip]{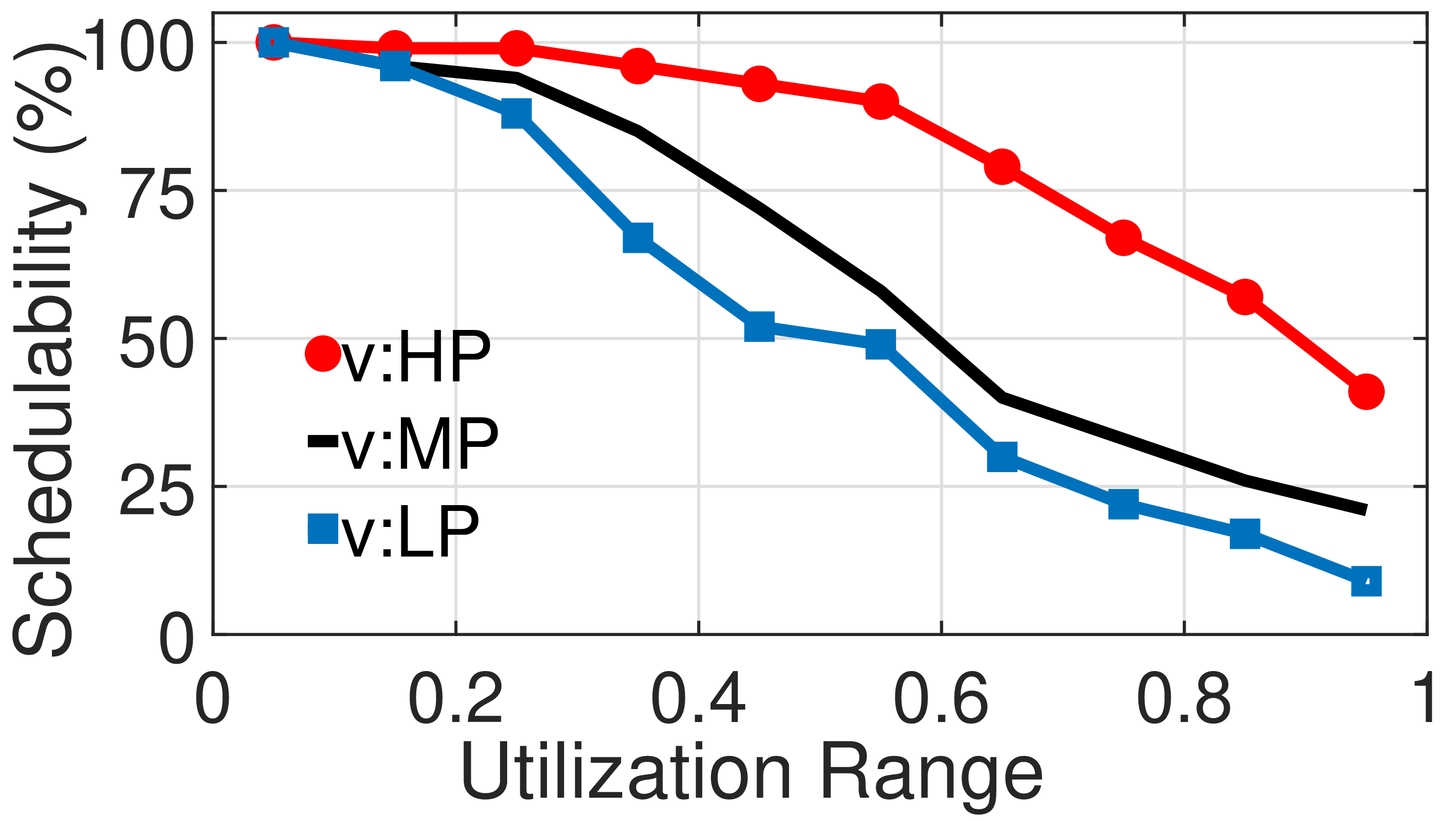}
    \caption{20 Tasks}
    \label{fig:schedulability_20}
    \vspace{-3mm}
\end{wrapfigure}
Eq.~\ref{eq:schedulability_victim} and compute WCRT of non-victim lower priority tasks $lp(\tau_v)$ using Eq.~\ref{eq:scheduability_non_victim}. We consider a task set to be schedulable if Eq.~\ref{eq:peakdelay} gives a positive peak job-level delay $\delta = \Delta_v^{peak} \in [0, T_v - C_v]$ which ensures the applicability of SecureRT. The schedulability percentage for all 100 task sets is computed separately, provided $\tau_v$ is selected from the HP, MP, and LP groups within every utilisation range. Fig.~\ref{fig:schedulability_5_10} shows the schedulability results for $n = 5, 10$, and Fig.~\ref{fig:schedulability_20} for $n= 20$ respectively. The $x$-axis denotes utilization groups, and the $y$-axis shows the percentage of task sets schedulable (out of 100) with $SecureRT$ given $\tau_v \in \text{HP}$ (red circular marked line), $\text{MP}$ (black solid line), and $\text{HP}$ (blue square marked line). We observe that schedulability drops significatly when $\tau_v \in HP$, followed by $\tau_v \in MP$ and $\tau_v \in LP$. Further, with an increasing utilization range and a higher number of tasks, it drops moderately. These results indicate that \emph{SecureRT} remains effective when $\tau_v$ is a relatively higher priority task in the set with low to moderate utilization. In modern CPS applications such as automotive systems, control tasks are usually designed to have higher priorities (shorter periods) and operate under moderate utilization, making \emph{SecureRT} suitable for most practical systems.
\section{Conclusion and Future Work}
\label{sec:conclusion}    
In this work, we propose a novel schedule-based attack mitigation framework, \emph{SecureRT}, that injects a bounded amount of job-level delay at each instance of safety-critical control tasks and minimises the time available to the untrusted task to launch an FDI attack. Unlike state-of-the-art approaches, SecureRT computes optimal delay injection sequences analyzing both schedulability and control performance.
As a natural future extension, we intend to make \emph{SecureRT} applicable to multi-core systems and for dynamic-priority schedulers such as EDF. We also aim to design adaptive job-level delay strategies that learn from system runtime behaviour to mitigate against intelligent schedule-based attackers.
\bibliography{reference.bib}
\bibliographystyle{ieeetr}
\end{document}